\documentclass[manuscript]{vldb}

\usepackage{graphicx}
\usepackage{balance}  

\usepackage{booktabs} 
\usepackage{cite}
\usepackage{xcolor}

\usepackage[english]{babel}
\usepackage[utf8x]{inputenc}
\usepackage[T1]{fontenc}

\usepackage{amsmath}
\usepackage{graphicx}
\usepackage[colorinlistoftodos]{todonotes}

\usepackage{minipage-marginpar}
\usepackage{caption}

\usepackage{pgfplots}
\usepackage{tikz}

\usepackage{calc}
\usepackage[english]{babel}

\usepackage{subcaption}

\usepackage{booktabs}
\usepackage{multirow}
\usepackage{float}

\usepackage{algorithm, algorithmic,xpatch}
\usepackage[algo2e,algosection, boxruled, linesnumbered]{algorithm2e}

\newcommand{\Comment}[1]{{\hskip3em$\rightarrow$/* #1*/}}

\usepackage{multicol}

\pgfplotsset{width=10cm,compat=1.9, label style={font=\large},
                    tick label style={font=\large}  }
                    \pgfkeys{/pgf/number format/.cd,1000 sep={\,}}

\newcommand{\seed}{\textit{MSSD}}
\newcommand{\un}{\textit{SI}}

\newcommand{\general}{\textit{MSSD}}
\newcommand{\den}{\textit{MSSD-D}}
\newcommand{\nn}{\textit{MSSD-N}} 
\newcommand{\rec}{\textit{MSSD-R}}
\newcommand{\opt}{\textit{MSSD*}}
\newcommand{\fix}{\textit{MSSD-F}}

\newtheorem{definition}{Definition}

\newtheorem{theorem}{Theorem}

\newtheorem{assumption}{Assumption}

\DeclareCaptionFormat{algor}{%
  \hrulefill\par\offinterlineskip\vskip1pt%
  \textbf{#1#2}#3\offinterlineskip\hrulefill}
\DeclareCaptionStyle{algori}{singlelinecheck=off,format=algor,labelsep=space}
\newcounter{parentalgorithm}

\makeatletter

\newenvironment{subalgorithms}{%

  \refstepcounter{algorithm}%

  \protected@edef\theparentalgorithm{\thealgorithm}%
  \setcounter{parentalgorithm}{\value{algorithm}}%

  \setcounter{algorithm}{-1}%
  \def\thealgorithm{\theparentalgorithm\alph{algorithm}}%
  \ignorespaces
}{%
  \setcounter{algorithm}{\value{parentalgorithm}}%
  \ignorespacesafterend
}
\makeatother

\makeatletter
\xpatchcmd{\algorithmic}
  {\arabic{ALC@line}}
  {\theALC@line}
  {}{}
\xpatchcmd{\algorithmic}{1.2em}{1.5em}{}{}

\newcommand{\updatelinenoprint}{%
  \setcounter{parentcounter}{\value{ALC@line}}
  \setcounter{ALC@line}{0}
  \renewcommand{\theALC@line}{\theparentcounter.\alph{ALC@line}}
}

\renewcommand{\theALC@line}{\arabic{ALC@line}}
\newcounter{parentcounter}

\usepackage{amsmath,amssymb,amsfonts}
\usepackage{algorithmic}
\usepackage{graphicx}
\usepackage{textcomp}
\usepackage{xcolor}

\makeatletter
\def\@copyrightspace{\relax}
\makeatother

\begin{document}


\title{Seed-Driven Geo-Social Data Extraction}



%
%
%
%

\numberofauthors{2} 

\author{
\alignauthor
Suela Isaj\\
 \affaddr{Department of Computer Science}\\
       \affaddr{Aalborg University}\\
       \email{suela@cs.aau.dk}
\alignauthor
Torben Bach Pedersen\\
 \affaddr{Department of Computer Science}\\
       \affaddr{Aalborg University}\\
       \email{tbp@cs.aau.dk}}

\maketitle

\begin{abstract}
Geo-social data has been an attractive source for a variety of problems such as mining mobility patterns, link prediction, location recommendation, and influence maximization. However, new geo-social data is increasingly unavailable and suffers several limitations. In this paper, we aim to remedy the problem of effective data extraction from geo-social data sources. We first identify and categorize the limitations of extracting geo-social data. In order to overcome the limitations, we propose a novel seed-driven approach that uses the points of one source as the seed to feed as queries for the others. We additionally handle differences between, and dynamics within the sources by proposing three variants for optimizing search radius. Furthermore, we provide an optimization based on recursive clustering to minimize the number of requests and an adaptive procedure to learn the specific data distribution of each source. Our comprehensive experiments with six popular sources show that our seed-driven approach yields 14.3 times more data overall, while our request-optimized algorithm retrieves up to 95\% of the data with less than 16\% of the requests. Thus, our proposed seed-driven approach set new standards for effective and efficient extraction of geo-social data.
\end{abstract}

\section{Introduction}
Each year social networks experience a continuous growth of 13\% in the number of users (http://wearesocial.com/uk /blog/2018/01/ global-digital-report-2018). Consequently, more and more information is available regarding the activity that the users share, events in which they participate and the new connections they make.  When data collected by social networks contain social connections (friendship links, mentions and tags in posts etc) as well as geographic information (check-ins, geo-data in posts and implicit location detection), then this data is usually referred as \emph{geo-social data}. \emph{Geo-social data} have attracted studies regarding location prediction, location recommendation, location-based advertisement, urban behavior etc. 

The primary sources of geo-social data are \emph{location-based social networks} (LBSNs) such as Gowalla, Brightkite, and Foursquare, which contain social ties, check-ins, tips and detailed information about locations. However, Gowalla and Brightkite were closed in 2012 whereas Foursquare has blocked the extraction of check-ins from its API (Application Programming Interface - set of functions and procedures that allow data extraction from a source). Other secondary sources of geo-social data are \emph{social networks} such as Facebook, Twitter, Flickr etc. Social networks are characterized by richness and variety of data, making them an attractive source for data extraction. However, the percentage of geo-located posts reported in the literature is less than 1\% (\cite{cheng2010you,jurgens2013s,li2013spatial,morstatter2013sample}). Furthermore, they provide rich information about users, their networks, their activities but only few details about locations (only the coordinates).  Another less common source of location data (not necessarily geo-social) are \emph{directories} such as Yelp, Google Places, TripAdvisor etc, which contain  locations with details such as name, phone, type of business, etc and sometimes accompanied by user reviews. In the majority of the cases, directories do not contain user profiles; even when they do, the API does not provide functions to extract user's information. To sum up, in order to obtain geo-social data, it is necessary to use several sources in order to gain a complete dataset.

Not only is geo-social data scattered over several sources but the APIs of the sources are also highly restrictive.  The APIs have restrictions regarding the number of requests that can be made within a time frame, the amount and the type of data that can be extracted, etc. The data extraction is a time-consuming procedure that might need months to get a relevant dataset. Instead of extracting the data, publicly available datasets can be used.
However, these datasets sometimes lack the details about users' profiles or the locations, and the check-ins/photos/posts/reviews might be sparse and scattered all over the globe. 
Enriching these datasets with the missing details is not possible because the data is anonymized so the link with the source is lost. Even when the data is not anonymized, the datasets are old (2008-2013) and they can not map to the existing users or locations of nowadays. When we analyzed 32 papers from 2009 to 2018 using geo-social data, we found that no less than 50\%  used datasets that  are 3-8 years older than the published article (see Section A in Appendix).
While in some cases these datasets can be quite suitable for the purpose, e.g. when the only structure of the network is analyzed, in other cases, they might not be of help. For example, some recommender systems and location prediction works need details about the users (age, location, posts) and the locations (semantics, rating). When the research is related to frequent patterns, mobility patterns, urban behavior, the sparsity of the activity of users such as check-ins or geo-tagged photos and texts affects the quality of the experiments.

To sum up, geo-social data is becoming even more needed and even less accessible. We thus, address the problem of \emph{location-based geo-social data} extraction from social networks and location-based sources. We introduce and quantify the limitations of six sources of geo-social data, namely: Flickr, Twitter, Foursquare, Google Places, Yelp, and Krak. Then, we propose a seed-driven algorithm that uses the points of the richest source (the \emph{seed}) to extract data from the others. Later, we introduce techniques to adapt the radius of the query to the searched area and to minimize the number of requests by using a recursive clustering method. 

Our main contributions are: (i) We provide an in-depth analysis of the current limitations of data extraction from six popular geo-social data sources. (ii) We identify and formulate the problem of \emph{maximizing the extracted geo-social data} while \emph{minimizing the requests} to the sources. To the best of our knowledge, we are the first to optimize the data extraction process in social networks and location-based sources. (iii) We propose a novel algorithm for data extraction that uses the points of one source as \emph{seed} to the API requests of the others. Our \emph{seed-driven} algorithm retrieves up to 98 times more data than the default API querying. (iv) We introduce an \emph{optimized} version of our algorithm that minimizes the requests by clustering the points and ensures maximized data extraction by (i) a recursive splitting of clustering depending on the data distribution of the source and (ii) a recursive establishment of the search radius  depending on the density of the area. We are able to retrieve around 90\% of the data using less than 16\% of the requests.

The remainder of the paper is structured as follows: first, we describe the related work in Section 2; then, we introduce the definitions and the data extraction problem in Section 3; later, we categorize the limitations of the data extraction process and we provide preliminary results from six sources in Section 4; we continue with formalizing our proposed algorithm in Section 5; next, we test the proposed solutions through real-time querying of the sources and we compare the results in Section 6; and finally, we conclude and provide further insights on our work in Section 7.

\section{Related work}
\label{sec:relatedwork}
Despite the growing interest in geo-social related topics, the existing related work does not focus specifically on optimizing the data extraction process. Most of the existing research uses either publicly available datasets \cite{wang2013location,emrich2014geo,li2014efficient,yin2016discovering,cai2016using,saleem2017location,yu2017friend,li2017geo,zhu2017geo,zhang2013igslr,ference2013location,saleem2018effective,li2017geo,zhao2016stellar,feng2015personalized,yao2016poi}, crawl using the default settings of the API \cite{cho2011friendship,gao2013exploring,jurgens2015geolocation,li2009analysis,noulas2012random,zhang2012evaluating,weiler2015geo,scellato2011exploiting,burini2018urban} or do both \cite{ijcai2017-314,liu2017experimental}. The (sparsely described) crawling methods used in these papers can be categorized as either \emph{user-based} crawling or \emph{location-based} crawling. The former extracts data from users by crawling their network. The latter uses queries with a specific location. Another method of querying [not necessarily crawling] is \emph{keyword-based} querying, widely used by the research on topic mining, opinion mining, the reputation of entities, quality of samples and several related topics \cite{bennacer2017interpreting, kwak2010twitter, morstatter2013sample} but not in research on geo-social topics.

\textbf{User-based crawling}. 
User-based crawling is based on querying \emph{users} for their data and their networks as well. 
A user-based crawling technique mentioned in several studies is the \emph{Snowball} technique \cite{gao2015content,scellato2010distance}. \emph{Snowball} requires a prior \emph{seed} of users to start with and then, traverses the network while extracting data from the users considered in the current round. Nonetheless, \emph{Snowball} is biased to the high degree nodes \cite{lee2006statistical} and requires a well-selected seed. 
Another interesting method is to track the users that post with \emph{linked accounts} \cite{hristova2016measuring,armenatzoglou2013general,preoctiuc2013mining}, for instance, users posting from Twitter using the check-in feature of Foursquare. This method allows obtaining richer datasets than a single social network. Nevertheless, it is limited only to linked accounts (whose percentage is less than 1\%) and requires additional filtering if we are interested only in a specific area, resulting in wasted requests.

\textbf{Location-based crawling}. Location-based crawling requires no prior knowledge and the extraction process can start at any time. It is based on extracting data near or within a specific area. Lee et al. \cite{lee2010measuring} use a  periodical querying based on points extracted from Twitter. First, Twitter is queried for initial points. Then, in a later step, other requests are performed using the initial points as query points, focusing on areas detected by the user. Thus, in each step $n$, the points discovered in step $n-1$ are used to perform the new queries. We will refer to this method as \emph{Self-seed}.

\textbf{Keyword-based querying}. As the name suggest, the source is queried with a keyword to find relevant data. The keyword-based querying is not directly applicable for geo-located information since querying with a keyword does not guarantee that the retrieved data will be located in the location mentioned by the keyword. For example, querying Twitter with the keyword "Brussels" can return tweets in Brussels, tweets talking about Brussels but not located there, and even tweets about brussels sprouts. Moreover, this method needs a good selection of keywords regarding the area of interest (names of cities, municipalities etc), otherwise it is bound to retrieve scarce data.

\textbf{Discussion.} Obviously, the keyword-based querying is not of interest due to the noise it brings and thus, resulting in wasted requests. The user-based crawling requires prior information about a seed of users and is applicable only to social networks.  Subsequently, they leave aside other location-based sources such as Yelp, TripAdvisor and Google Places, which provide interesting information about the locations in geo-social data. Moreover, if the study is based on a region of interest, the user-based crawling results in a lot of irrelevant data because even if the seed of users is well-selected from the region of interest, there is no guarantee that the friends will contain check-ins in the area of interest. Consequently, user-based crawling produces wasted requests. The method described by \cite{lee2010measuring} has some similarities with ours because it is location-based crawling and focuses on performing requests on areas discovered previously. In comparison, our approach differs significantly because (i) instead of selecting points from a single source and querying itself, we use a geographically rich \emph{seed} to query \emph{multiple sources}, (ii) we minimize the number of requests performed while maximizing the data extracted  (iii) our seed-driven data extraction approach does not need periodical querying; it can be run continuously and simultaneously for all the sources, resulting in faster data extraction process, compared to several months like in \cite{lee2010measuring}. \emph{ To sum up, our data extraction approach is faster, richer, request-economic and includes multiple sources}.

\section{Problem definition}

Extracting geo-social data from a region of interest can be done by directly querying to get data in that region with location-based queries. The notion that we will use widely in the paper is the notion of a \emph{location}. A location in a directory is a venue with a geographical point and additional attributes like name, opening hours, category etc. However, social networks contain \emph{activities} such as check-ins, tips, photos and tweets which are geo-tagged. We denote the locations associated with the activities as \emph{derived locations}. For brevity, from now on, both locations from directories and derived locations from social networks are referred to as \emph{locations}.

\begin{definition}{
\label{def:activity}
A location $l$ is a  spatial entity identified within the source by a unique identifier $id(l)$. A location $l$ has a set of attributes $A=\{a_1, a_2 ... a_n\}$ accompanied by their values $\{a_1(l), a_2(l) ... a_n(l)\}$}.
\end{definition}

The id of a location ($id(l)$) is unique within the source. For example, $l_1$ is a tweet with id = 1234567, where $A$=\{text, user, point\} and the values are
\{"Nice day in park", 58302, <57.04, 9.91>\}. A required attribute for a location is its geographical coordinates denoted as $p(l)$.

Geo-social data sources usually offer an \emph{Application Programming Interface (API)}, which is a set of functions and procedures that allow accessing a source to obtain its data. Location-based API calls allow querying with (i) a point $p$ and a radius $r$, (ii) a box $<p_1, p_2, p_3, p_4>$ and (iii) keywords. We will not consider keyword-based querying due to the noise it bring (see Section \ref{sec:relatedwork}). 
The circular querying (i) and the rectangular querying (ii) are quite similar as long as the parameters are the same. We focus mostly on querying with a point $p$ and a radius $r$ due to the fact that most of our locations are identified by points. We will refer to the searched area as $Circle (p,r)$. We define a geo-social data source formally as:

\begin{definition}
A geo-social data source $S$, short as \emph{source}, consists of the (complete) set of locations $L(S)$ and a source-specific extraction function $API$: $\mathcal{P}x\mathcal{R}^+\Rightarrow 2^{L(S)}$, where $\mathcal{P}$ is the domain of geographical points and $\mathcal{R}^+$ is the domain of non-negative numbers. $API(p,r)$  queries with a centroid $p \in \mathcal{P}$ and a radius $r\in \mathcal{R}$ and returns a sample of locations $L_p^r$, such that for each $l \in L_p^r$, $p(l) \in Circle (p,r)$ and $|L_p^r|\leq M_{S}$, where $M_{S}$ is the maximal result size for  $S$. 
\end{definition}

For instance, if $S$ is Twitter, then $L(S)$ is the complete set of tweets  (all the tweets posted ever on Twitter). We can query Twitter using the API with a point and a radius $<p,r>$ and let $L_p^r$ be the result of the query. $L_p^r$ is a \emph{sample} of size at most $M_S$ of the underlying activities $L_p^r(S)$ in $Circle(p,r)$. So, if $M_S=100$, then the size of $L_p^r$ is not larger than 100. 
Let us now suppose that we have a budget of only $n$ API calls. These calls need to be used wisely in order to retrieve the largest combined result size. Moreover, each query result should contribute with new locations. For example, if the first request retrieves the locations $\{l_1, l_4, l_5, l_6\}$ and the second request retrieves $\{l_2, l_4, l_5, l_6\}$, then the second request only contributed with one new location ($l_2$).
Hence, we need to find which pair of point and radius $<p,r>$ would provide the largest sample with the highest number of new locations.

\emph{Problem definition}.
Optimizing geo-social data extraction is the problem that given a source $S_i$ and a maximum number of requests $n$ finds the sequence of pairs of point and radius $\{<p_1, r_1>, <p_2, r_2> ... <p_{n}, r_{n}>\}$ such that the size of $L_i=\bigcup\limits_{j=1}^{n} L_{p_j}^{r_j}$ is maximized .

The problem aims to obtain a good compromise between the number of requests $n$ and the number of locations $L_i$. Additionally, the problem intends to minimize the intersection of the results $\bigcap\limits_{j=1}^{n} L_{p_j}^{r_j}$ as well. The solution to our data extraction problem is a combination of $\{<p_1^*, r_1^*>, <p_2^*, r_2^*> ... <p_n^*, r_n^*>\}$ such that  $L_i$ is maximal. Let us denote as $L_i^*$ the result of the \emph{optimal solution} to our problem ($L_i^* \subseteq L(S_i)$).   There might be several combinations of points and radii that retrieve a specific number of locations. In order to find $L_i^*$, we have to try exhaustively all possible values and combinations of $p$ and $r$ and rank the retrieved results. Given that we have a budget of $n$ requests, finding the optimal pairs of $<p_j,r_j>$ is not feasible. Hence, we will have to propose solutions that are based on heuristics and assumptions.
But, before proposing our solutions, let us first study the limitations of the APIs for each of the sources.

\section{Limitations of existing \\ geo-social data sources}
 \label{sec:limitations}

\subsection{API Limitations}

\begin{table*}[htb]
\centering
\scriptsize
\begin{tabular}{@{}lllllll@{}}
\toprule
\textbf{API limitations}      & \textbf{Krak}               & \textbf{Yelp}   & \textbf{Google Places}  & \textbf{Foursquare}  & \textbf{Twitter}  & \textbf{Flickr}                  \\ \midrule
\textbf{Bandwidth}  & 10K in 1 month   &  5K in 1 day & 1 in 1 day (from 6/2018) & 550  in 1 hour & 180 in 15 min  & 3.6K in  1 hour 
 \\ \hline
 
\textbf{Max Result Size}& 100  & 50  & 20  & 50  & 100  & 500 
  \\ \hline
  
  \textbf{Historical Access} & N/A & N/A & N/A & Full & 2 weeks  & Full \\  
  
  \hline
  
  \textbf{Supp Results} & 4.3\% & 17.3\% & 0.5\% & 0.0\% & 0.0\% &  0.0\%
  \\ \hline
  
   \textbf{Complete access} & yes  & yes & yes & yes & 1\% & yes 
  \\ \hline
  
    \textbf{Cost} & not stated  & negotiable & from 200\$/month & from 599\$/month & 149\$ - 2499\$/month & not stated  
  \\ \hline

  \bottomrule
\end{tabular}

\caption{Summary of limitations of social networks}
\label{tab:categories}
\end{table*}

We introduce the API limitations, which indicate the technical limitations of extracting locations from geo-social sources and provide some analysis of the spatial and temporal density of the extracted data.

\emph{Data acquisition}: 
With regard to quantifying the limitations, we present preliminary results from querying six sources: Twitter, Flickr, Foursquare, Yelp, Google Places, and Krak.  Krak is a Danish website that offers information about companies, telephone numbers, etc. In addition, Krak is part of Eniro Danmark A / S. which takes care of publishing The Yellow Pages. We queried all the sources simultaneously for the region of North Denmark during November-December 2017. With respect to gaining more data, we performed additional requests using different keywords ("restaurant", "library", "cozy", etc) as well as coordinates of the cities and towns in the region.

\begin{figure*}[htb]
 \centering 
 \begin{minipage}{.35\textwidth} 
 \centering  \includegraphics[width=\textwidth]{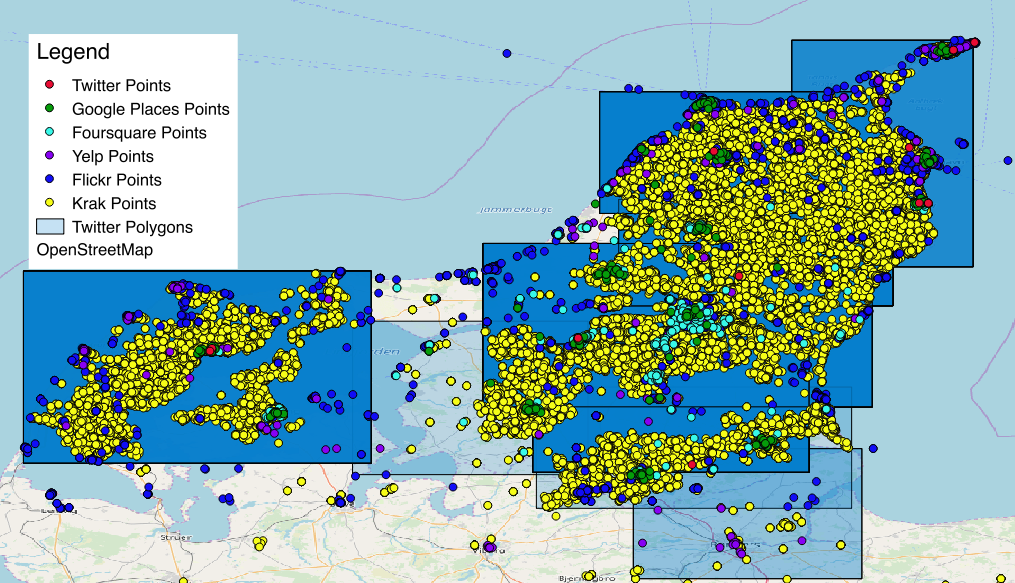}
    \caption{Geographical spread of locations}
    \label{fig:map}
 \end{minipage}%
 \begin{minipage}{.65\textwidth}
\begin{subfigure}[b]{0.33\linewidth}
    \centering
  \includegraphics[width=\linewidth]{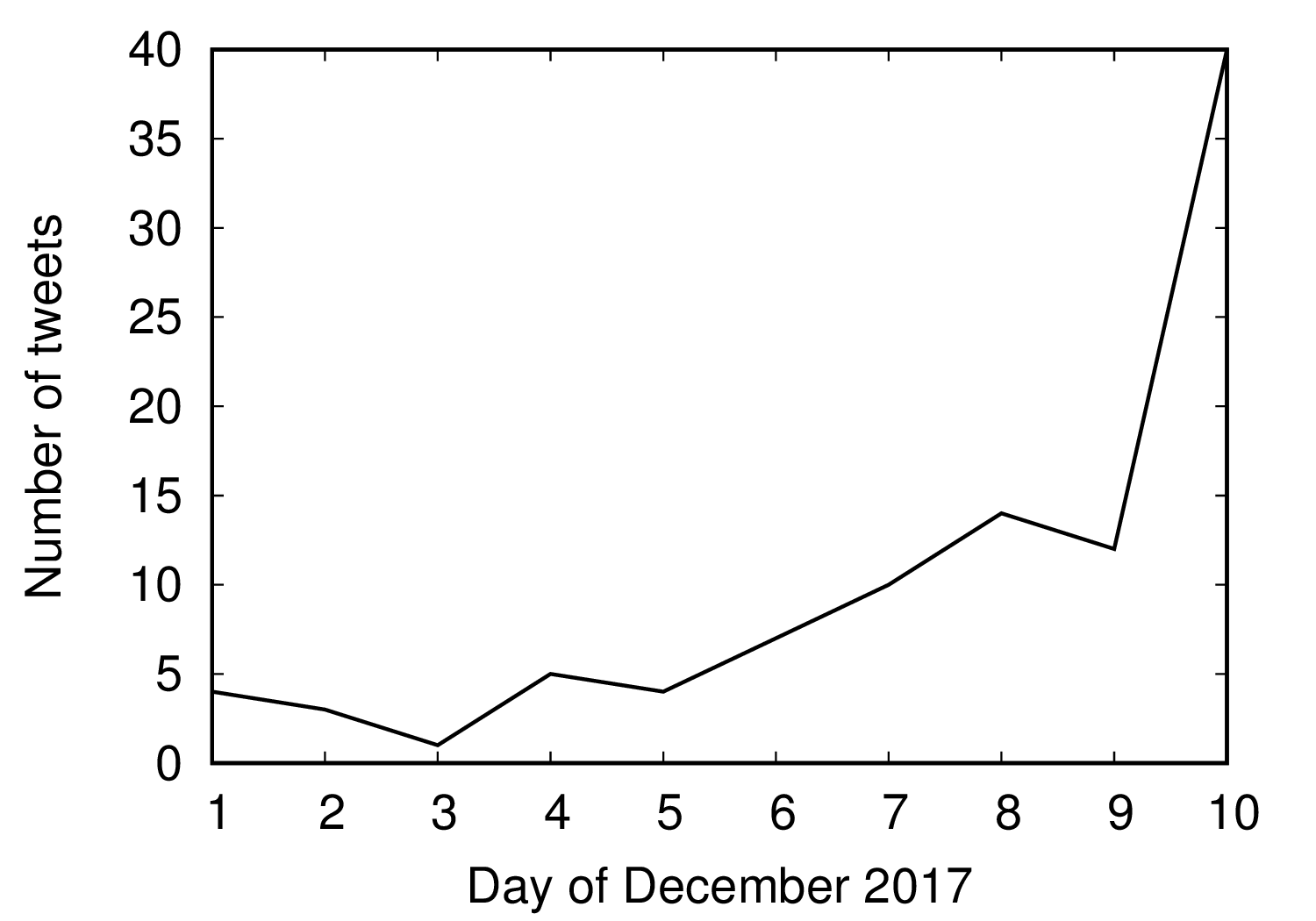}
    \caption{Twitter} 
    \label{fig:temporalTwitter}
    \end{subfigure}
  \begin{subfigure}[b]{0.33\linewidth}
    \centering
 \includegraphics[width=\linewidth]{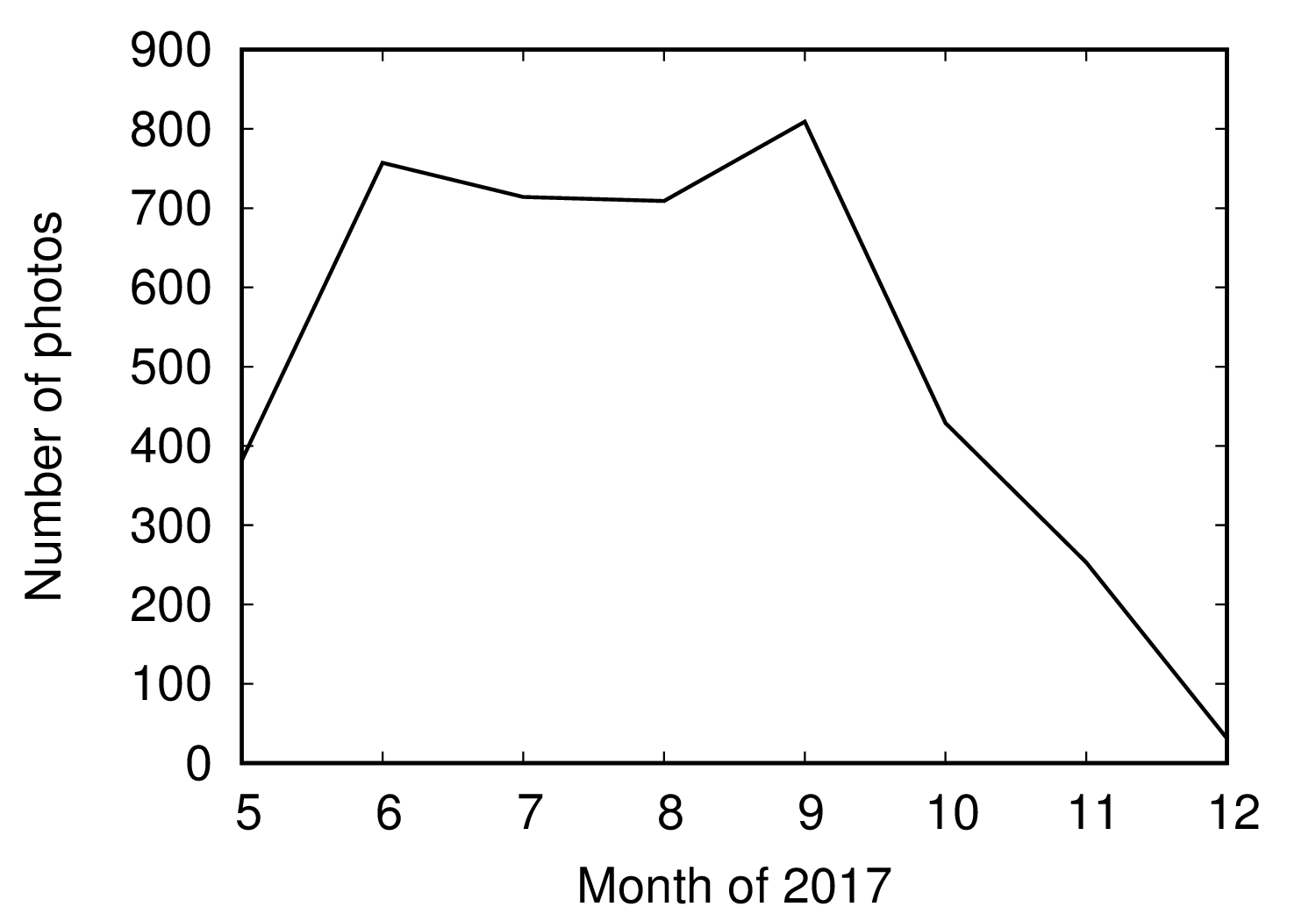}
    \caption{Flickr} 
   \label{fig:temporalFlickr}
   \end{subfigure} 
  \begin{subfigure}[b]{0.33\linewidth}
    \centering
    \includegraphics[width=\linewidth]{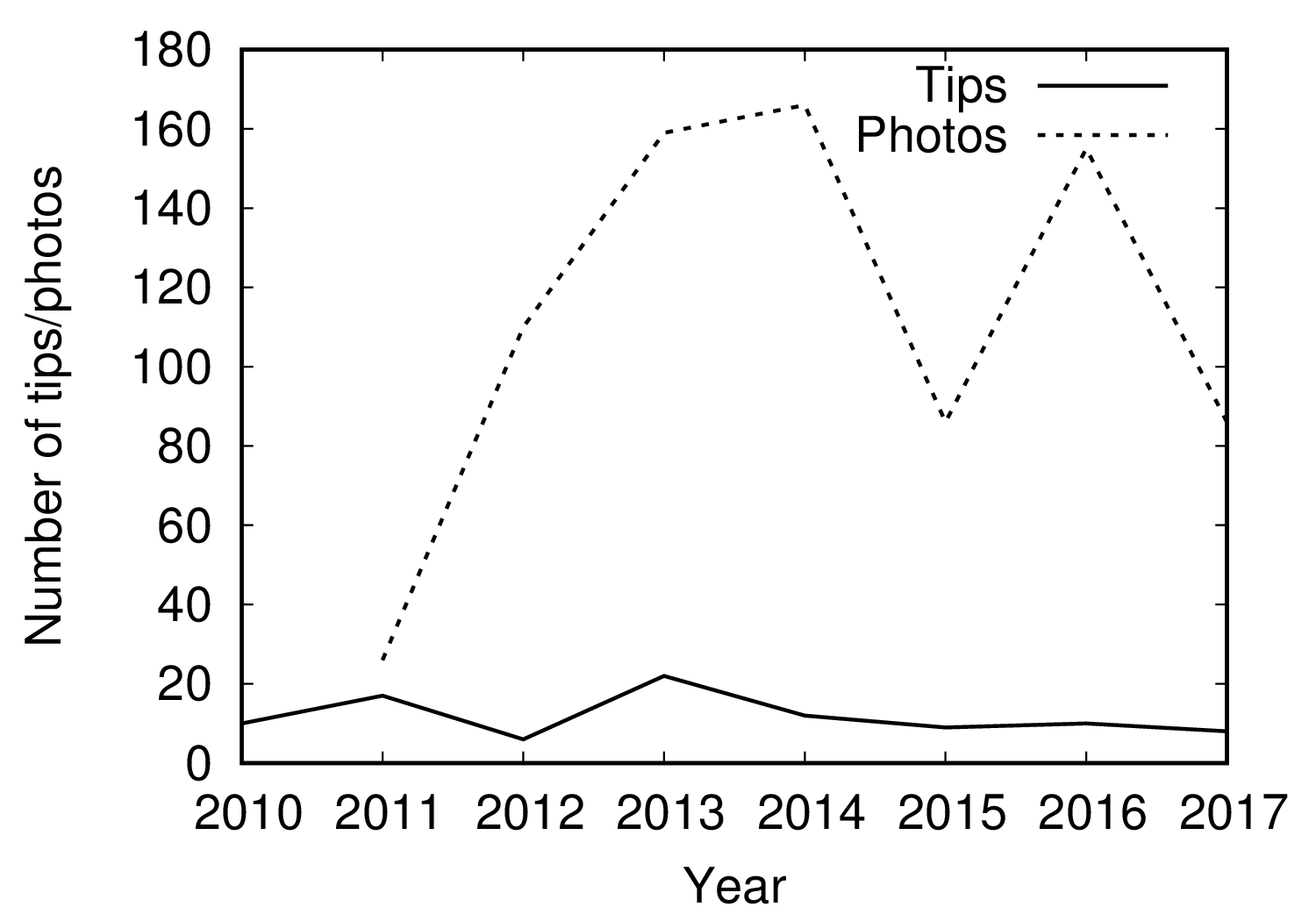}
      \caption{Foursquare} 
   \label{fig:temporalFSQ}
   
  \end{subfigure} 

  \caption{Temporal sparsity} 
  \label{fig::temporalALL}

 \end{minipage}
 \end{figure*}

\textbf{Bandwidth}. The API bandwidth refers to the number of requests that can be performed within a time frame. API bandwidth differs from one social network to the other also for the granularity of time in which the requests are performed. For example, Twitter allows 180 requests in 15 min, meanwhile, Krak has a time window of a month. Google Places allowed 1000 requests in a day before June 2018 and now, just one request per day. If more requests are needed, the cost is 0.7 US cents/request (first 200 USD free). In our data extraction and experiments, we fix the bandwidth of Google Places to 1000 requests in a day (the former default). 

\textbf{Maximal result size}. The maximal result size is the maximal number of results returned by a single request. For instance, an API call in Flickr retrieves 250 photos but an API call in Google Places retrieves only 20 places.

\textbf{Historical access}.
This characteristic of the API is related to how accessible the earlier activity is. Directories such as Yelp, Krak and Google Places do not provide historical data; they only keep track of the current state of their locations. Foursquare provides only the current state of its venues and historical access to photos and tips by querying with venues. 
Flickr is able to go back in old photos whereas Twitter limits the results only to the last couple of weeks. 

\textbf{Supplemental results}.
Sometimes there is no data in the queried region. So, some APIs might return an empty result set, others might suggest the nearest entity matching the query. For example, if we search for "Zara" shop in city X, the API might return the "Zara" shop which is the closest to X but in city Y. We name these results as \emph{supplemental} as they are provided in addition to what we search for, in order to complete the request. We noticed that supplemental results are present only in directories ("Supp results" in Table \ref{tab:categories}), which aim to advertise and provide results anytime.

\textbf{Access to the complete dataset}. This feature shows if the API can query the whole dataset or just a sample. For example, the Twitter API accesses only a sample of 1\% while others enable access to its complete dataset.

\textbf{Costs}. All social networks and directories offer free APIs at no cost but with the aforementioned limitations. They also offer premium or enterprise services with a monthly payment or pay-as-you-go services. While some sources have predefined pricing plans (Twitter, Foursquare and Google Places), others offer the possibility to discuss the needs and the price (Yelp). Even though a premium service has less restrictions, it is still needed to keep costs down.

 \emph{Summary}. A summary of the limitations of social networks is presented in Table. \ref{tab:categories}. Krak shows promising results in terms of richness but is restrictive with the bandwidth. Google Places has a very small result size and only 1 request per day. Flickr is promising in terms of the API limitations, while Twitter shows severe problems regarding the limitation to access historical data. Foursquare and Yelp could be considered similar in terms or limitations.

\subsection{Data Scarcity}
\label{sec:datascarcity}
In this section, we analyze the quality of the extracted data from all six sources.

\textbf{Geographical sparsity}. The geographical sparsity indicates whether the number of locations available in a region is low or high. In order to have a visual perception of the spread of the points, we plotted them on the map (Fig. \ref{fig:map}). Twitter polygons cover a big area (some of them even reach 12,5\% of the whole region), while in terms of points, Twitter is the sparser. Google Places and Foursquare are both sparse but the former is spread more uniformly than the latter. Yelp and Flickr have points on the coast, which is a touristic area. With regard to the number of locations and number of points (Table \ref{tab:nrpoints}), it is possible to distinguish the leading position of Krak, which has almost 2 orders of magnitude more results than any other source, followed by Flickr, Foursquare, Yelp, Google Places and finally, Twitter. Moreover, Krak covers locations from 2.308 categories, including restaurants, bars, attractions along with locations that are related to everyday activities such as hair salon, fitness center, hospitals, kindergartens etc.


\begin{table}[htbp]
\centering
\scriptsize
\begin{tabular}{@{}lllllll@{}}
\toprule
\textbf{Category}      & \textbf{Krak}               & \textbf{Yelp}   & \textbf{GP}  & \textbf{FSQ}  & \textbf{TW}  & \textbf{Flickr}                  \\ \midrule
\textbf{Locations}  & 143,073  &  473 & 380 & 1,097 & 115  & 4,084 \\
\textbf{Points}  & 32,461  &  467 & 356 & 1,093 & 25  & 2,272
\\
\bottomrule
\end{tabular}
\caption{Number of locations and  points for each source}
\label{tab:nrpoints}
 \end{table}
 
\textbf{Temporal sparsity}. Since directories keep only the current state of their data, only social networks will be analyzed from this perspective. We recorded that the average of photos per day in Flickr is around 17, while for Twitter is 10. Foursquare is very sparse the average would be 0.03 tips a day, while for photos it would be 0.36 photos per day. The distribution over time is more interesting (Fig \ref{fig:temporalTwitter}, Fig \ref{fig:temporalFlickr}, Fig. \ref{fig:temporalFSQ}). The sources have considerable differences in their temporal activity. Flickr API has a larger result size and can retrieve historical data. Twitter can access only the latest couple of weeks.  Flickr has a high usage in the summertime as the coast of the region is touristic while Foursquare is generally quite sparse.

   



\emph{Summary.} As shown in Section \ref{sec:datascarcity}, a single source cannot provide a rich enough dataset. Given the API limitations, the requests need to be managed in order to improve the density of the data. Since we have no prior knowledge regarding the queried area and the distribution of the points, this knowledge can be obtained by using multiple sources operating in the area. In the next section, we propose a novel algorithm that uses one of the sources as \emph{seed} to extract data from the others and is capable of obtaining up to 14.3 times more data than single source initial querying (see Section \ref{sec:experiments}).

\section{Multi-source Seed-Driven\\ approach}

Section \ref{sec:limitations} studied the limitations of data extraction and quantified the performance of each of the sources. In this section, we propose a main algorithm and several adaptions to it that lead to an effective data extraction process.

\subsection{Multi-Source Seed-Driven Algorithm}


As shown in the previous section, when the user has no prior and in-depth knowledge of the distribution of the locations in the source, querying the source with the default API settings and some user-defined API queries cannot achieve a rich dataset. We will denote the initial locations  obtained from each source as \un. 

\textbf{Source Initial} (\un) is the set of locations $L_I$ retrieved from the $k$ initial requests with default API parameters. 



Having no prior knowledge of the underlying data $L(S_i)$ makes it difficult to choose which API calls to perform. However, \emph{all the sources operating in the same region contain data that maps to the same physical world}. For example, if there is a bar in the point (56.89 9.21) in Krak, probably around this point there might be this and other locations in Yelp, Google Places, and Foursquare and even some activity such as tweets, photos or check-ins in Twitter, Flickr, and Foursquare. This means that if the \un\ of a source is rich, then its knowledge can be used to improve the data extraction of the other sources. Hence,  we propose a \emph{seed-driven} approach to extract locations from multiple sources. The main idea is simple; \emph{selecting one (more complete) source as the seed and feeding the points to the rest for data extraction}.

\textbf{Multi-Source Seed-Driven}  (\general) is a function that takes as input a set of sources $S_1, S_2...S_k$ and outputs their corresponding sets of locations  $\{L_{S_1}, L_{S_2}, ... L_{S_k} \}$ obtained from the seed-driven approach in Alg. \ref{alg:seed}.

For example, let us suppose that the seed provides a location with coordinates (57.05, 9.92) as in Fig. \ref{fig:ideaseeddriven}. We can search for locations across sources within the circle with center (57.05, 9.92) and a predefined radius. The different colors in the figure represent the different sources. We can discover three locations from the red source, two from the blue source and two from the green source. The algorithm for the seed-driven approach is presented in Alg. \ref{alg:seed}. Selecting a good seed is important, thus we start by getting the most complete source (with the most points) in line 4. The points in the seed indicate regions of interest and are used for the API request in the sources. So, for each point in seed (for each $p$ in $P$), we query  the rest of the sources except the seed source. Line 7 shows the general API call for each of the sources, which actually is performed in correspondence to the requirements of the source. The request takes the coordinates of $p$ and the radius $r$. The search returns a set of locations $L_p^r$, which is unioned to our source-specific output $L_S$. From Alg. \ref{alg:seed},  $L_S = \bigcup\limits_{j=1}^{k}( \bigcup\limits_{p=1}^{∣P∣}( L_p^r \bigcup L_I )) $ so $ L_I \subseteq L_S$. 

\begin{figure*}[t]
    \centering
 \begin{minipage}{.23\textwidth}
  \centering
\includegraphics[width=0.85\textwidth]{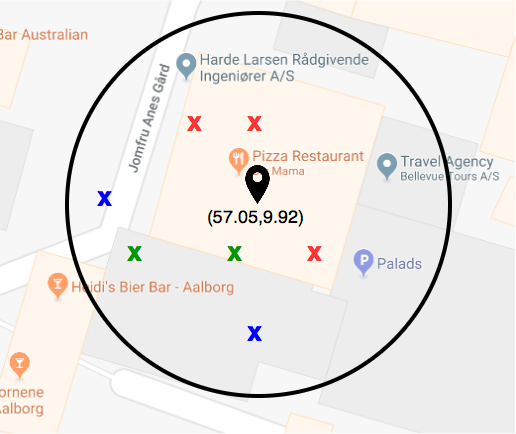}
    \caption{\seed\ approach}
    \label{fig:ideaseeddriven}
\end{minipage}%
\begin{minipage}{.75\textwidth}
  \centering
 \begin{subfigure}[b]{0.31\linewidth}
    \centering   \includegraphics[width=0.9\linewidth]{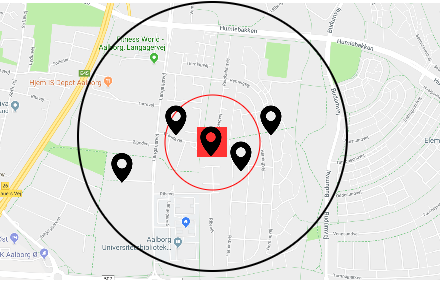}
     \caption{\den\ radius}
    \label{fig:ideasdensitybased}
  \end{subfigure}
  \begin{subfigure}[b]{0.31\linewidth}
    \centering    \includegraphics[width=0.9\linewidth]{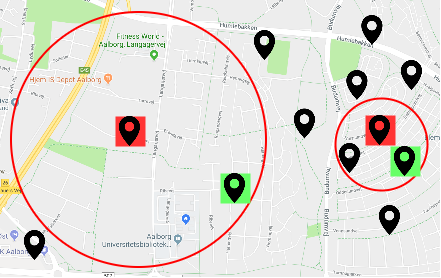}
    \caption{\nn\ radius}
    \label{fig:ideasneighbor}
  \end{subfigure}\begin{subfigure}[b]{0.31\linewidth}
    \centering    \includegraphics[width=0.9\linewidth]{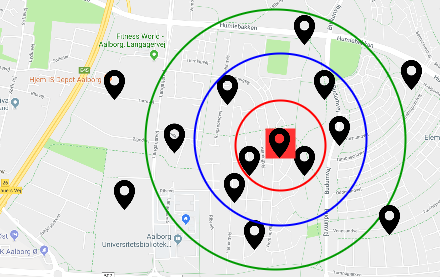}
    \caption{\rec\ radius}
    \label{fig:idearecursive}
  \end{subfigure}
  \caption{\seed\ radius} 
  \label{fig:radius}
\end{minipage}
\end{figure*}

\begin{subalgorithms}

\begin{algorithm} [htp]
\caption{Multi-Source Seed-Driven (\seed) } 
\label{alg:seed} 
\begin{algorithmic}[1]
\INPUT A set of sources $\{S_1, S_2, ...S_k\}$, radius $r$
\OUTPUT  $\{L_{S_1}, L_{S_2}, ... L_{S_k} \}$
\FOR{\textbf{each} $S$ in $\{S_1, S_2, ...S_k\}$}
\STATE $L_S\gets L_I$ \Comment{Initialize $L_S$ of each source with $L_I$}
\ENDFOR
\STATE  Let $S_{seed}$ be the source with the most points in $\{S_1, S_2, ...S_k\}$, $L_{seed}$ its locations  and $P$ the distinct points in  $L_{seed}$                   
\FOR{\textbf{each} $p$ in $P$}  

\FOR{\textbf{each} $S$ in $\{S_1, S_2, ...S_k\}$ - $S_{seed}$}
\STATE $L_p^r \gets API(p, r)$ \Comment{API request for the source $S$ } 
\STATE  $L_S\gets L_S\cup L_p^r$
\ENDFOR

\ENDFOR

\RETURN{\  $\{L_{S_1}, L_{S_2}, ... L_{S_k} \}$}
\end{algorithmic}
\end{algorithm}


\subsection{Optimizing the Radius}
\label{sec:radius}

\begin{figure}[b]
\centering
 \begin{subfigure}[b]{0.35\linewidth}
    \centering   \includegraphics[width=0.95\linewidth]{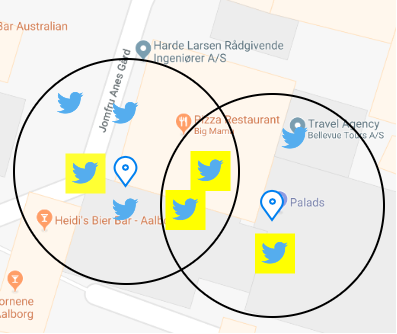}
     
  \end{subfigure}
  \begin{subfigure}[b]{0.35\linewidth}
    \centering    \includegraphics[width=0.95\linewidth]{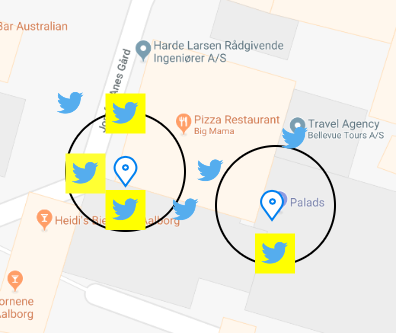}
   
  \end{subfigure}
\caption{Radius adjustment}
    \label{fig:radiusillustration}
  \end{figure}

We can improve further \seed\  by adapting the API request to the source. Alg. \ref{alg:seed} considers the radius as a constant. The choice of a good $r$ is a challenging task, given that we do not know the real distribution of locations in the sources. Even though a big radius might seem like a better solution because we query a larger area, note that the API retrieves only a sample of the underlying data. Hence, if we query with points that are nearby, we might retrieve intersecting samples. Fig. \ref{fig:radiusillustration} shows an example of the radius adjustment problem in Twitter. Let us assume that the result size is 3, which means that the API can not retrieve more than 3 tweets. If we query with a big radius as in the left part of Fig.\ref{fig:radiusillustration}, we might get only 2 out of 3 tweets in the intersection. If the radius is small, then we explore better the dense areas but we might miss in sparser ones like in the right part of Fig. \ref{fig:radiusillustration}. The union of tweets in both searches is just 4, where ideally it should have been 6. We propose two improvements: \emph{using the knowledge of the seed to query the sources} and  \emph{using a recursive method to learn a suitable radius for the source}. In the first improvement, we assume that the data distribution in the seed is similar to the sources we will query, so we adjust the radius accordingly.  In the second improvement, we use the points of the seed but we adjust the radius while querying the source.

\textbf{Multi-Source Seed-Driven Density-Based} \den. The radius in this version is defined by the density of points in the \emph{seed}.  Before the API requests, we check the density of points in the search area in the \emph{seed}. If the area is dense, then we reduce the radius accordingly. For example, we need to query $S_2$ and the seed is $S_3$. First, we check how many points in the seed $S_3$ fall within the circle with center $p$ and initial radius $r$ (no API requests needed). If there are $N$ points, then we query with $<p,r_d>$ where $r_d=\frac{r}{N}$.
\begin{algorithm} [htp]
\caption{ \seed\ Density-Based (\den) } 
\label{alg:densitybased} 

\begin{algorithmic}[1]

\setcounter{ALC@line}{5}


\updatelinenoprint
    
\STATE  Find $N=\{q∣q\in Circle (p, r)\}$ \Comment{Find how dense the region is}

\STATE $r_d=\frac{r}{∣N∣}$ \Comment{Adjust the radius}

\setcounter{ALC@line}{7}


\end{algorithmic}

\end{algorithm}

Fig. \ref{fig:ideasdensitybased} illustrates the intuition behind \den. We are using the point in red as seed point $p$. First, we check how many points of the \emph{seed} are in the search area (4 points in the black circle). Second, we adjust the radius according to the density, so in this case, we divide the radius by 4. Finally, we perform the API call to the source with the red circle. Alg. \ref{alg:densitybased} shows the alterations we make in Alg. \ref{alg:seed} for the radius calculation. We add line 5.a and 5.b after line 5 in Alg. \ref{alg:seed}. First, we find the density of the region and then, we adjust the radius depending on the density. We query with the adjusted $r_d = \frac{r}{\vert N \vert}$ in line 7 of Alg. \ref{alg:seed}.

\textbf{Multi-Source Seed-Driven Nearest Neighbor} \nn.
As the name suggests, we use the nearest neighbor in the \emph{seed} to define the radius. This approach guarantees that dense regions will be queried with a small radius and sparse regions with a large radius. For example, we need to query $S_2$ and the seed is $S_3$ as previously. First, we check the nearest neighbor of $p$ in $S_3$ (no API call needed), which is $q$, then we query with $<p,r_n >$ where $r_n=| p-q |$. 
A simple illustration of this idea is presented in Fig. \ref{fig:ideasneighbor}. For each of the points in the seed (in red), we find the nearest neighbor in the seed (in green) and then we query with the adjusted radius. Note that we query with a small radius in dense areas and a big radius in sparse ones. Alg. \ref{alg:nearestneighbor} instead adds line 5.a and 5.b after line 5 of Alg. \ref{alg:seed}. It finds the nearest neighbor $q$ of the point $p$. Then, we set $r_n=| p-q |$. The adjusted $r_n$ is used to query the sources in line 7 of Alg. \ref{alg:seed}.

\begin{algorithm} [htp]

\caption{ \seed\ Nearest Neighbor (\nn)} 
\label{alg:nearestneighbor} 

\begin{algorithmic}[1]

\setcounter{ALC@line}{5}

\updatelinenoprint

\STATE Find $q= \underset{q\in L_i}{\text{min}}| p-q |$ \Comment{Find the nearest neighbor}

\STATE $r_n=| p-q |$ \Comment{Adjust the radius}


\end{algorithmic}

\end{algorithm} 
\textbf{Multi-Source Seed-Driven Recursive} \rec. The previous versions (\den\ and \nn) rely on the hypotheses that the distribution of the points in the seed indicates similar distribution on the sources. For example, if there is a dense area in the seed, e.g. the city center, then most probably the same area is dense in the source. The advantage of this hypothesis is that no API call is needed to adjust the radius because these calculations are performed on the \emph{seed} points. Ideally, this hypothesis would hold but in fact, taking into consideration the analytics performed in previous sections, there is a difference between the distribution of points in the sources, originating mostly from their scope. Consequently, the seed-oriented methods might perform well in most of the cases but will fail in specific areas. For instance, if Krak is the seed and it is dense in the city center, that is a good indicator and most probably all the sources are dense in that area. But Krak will miss the density in the (touristic) coastline that Flickr reveals. Thus, there is a need for a better approach to assigning a good radius for a specific area. We propose a solution that adjusts the radius while querying the source. First, if an area not dense, we are able to identify it from the API call. We assume that if the area contains less than $M_S$ locations, the API call will retrieve all of them.

\begin{assumption}
For each source $S$ in $\{S_1, S_2, ..., S_k\}$, if $Circle <p,r>$ contains $L_p^r(S)$ locations such that $\vert L_p^r(S) \vert \leq M_S$, then $API(p,r)$ will retrieve $L_p^r = L_p^r(S)$.
\label{assumption:queryM}
\end{assumption}

The API retrieves a sample of size $M_S$  of the underlying data in a queried region $Circle (p,r)$. If the underlying locations $L_p^r(S)$ are already less than $M_S$, then we assume that the API will retrieve \emph{all} the locations lying in $Circle (p,r)$. For example, if there are 30 locations in $Circle ((56.78$ $ 9.67),1km)$ and $M_S=50$, then querying with $p=(56.78$  $9.67)$ and $r$=1 km will return all 30 locations.

\begin{theorem}
Let $<p,r>$ be a pair of point and radius such that $API(p,r) = L_p^r$ where $\vert L_p^r \vert  < M_S$. Then, for all $r'$ such that $r' < r$,  $ L^{r\prime}_p \subseteq  L_p^r $. 
\label{theorem:stoprecursive}
\end{theorem}

\begin{proof}
Let us assume that there are  $\vert L_p^r (S)\vert $ locations in $Circle(p,r)$  and $\vert L_p^{r\prime} (S)\vert $ locations in $Circle(p,r\prime)$ . Since $r'<r$, then the surface covered by $Circle(p,r\prime)$ is smaller than the surface covered by $Circle(p,r)$ ($\pi {r'}^2 < \pi {r}^2$). Consequently, $ L_p^{r\prime} (S) \subseteq L_p^{r} (S)$. According to Assumption \ref{assumption:queryM}, since $API (p,r)$ retrieves $ \vert L_p^r \vert < M_S$, then $ L_p^r  = L_p^{r} (S) $ and $ \vert L_p^r (S) \vert < M_S$. Given that $ L_p^{r\prime} (S) \subseteq L_p^{r} (S)$ and  $ \vert L_p^r (S) \vert < M_S$, we conclude that $ \vert L_p^{r\prime}\vert < M_S $  and $L_p^{r\prime}  = L_p^{r\prime} (S) $. Finally, from $ L_p^{r\prime} (S) \subseteq L_p^{r} (S) $, we derive that $  L_p^{r\prime} \subseteq  L_p^{r} $.
\end{proof}

This is an important finding that will be used in defining \rec. Since there are less than $M_S$ locations retrieved by the API call in source $S$, there are no new locations to be gained by querying with a smaller radius.
Thus, we propose a recursive method to assign the radius that uses Theorem \ref{theorem:stoprecursive} as our stopping condition. First, we query with an initial large radius and if the result size is $M_S$, then we know this is a dense area and we perform another request with a smaller radius. The search stops when the number of returned results is smaller than the maximal result size because according to  Theorem \ref{theorem:stoprecursive}, there is no gain in reducing $r$ further. For example, we need to query $S_2$ and the seed is $S_3$. Let us suppose that the maximal result size is 50 locations. If the request of $<p,r>$ returns 30 locations, then this is not a dense area. Otherwise, we perform another request with $<p,r_r>$ where $r_r=\frac{r}{\alpha}$ and $\alpha$ is a coefficient for reducing the radius. We stop when we retrieve less than 50 locations. The recursive method is illustrated in Fig. \ref{fig:idearecursive}. Let us suppose that we are querying source $S$ and the maximal result size is $M_S=5$. After querying with the green circle, we get 5 locations so we know that the area is dense. We reduce the radius by $\alpha$ and query again with the blue circle. We get again 5 locations, therefore we continue once more with a smaller radius. When we query with the red circle, we get only 2 locations so we stop.

\begin{algorithm} [htp]

\caption{\seed\ Recursive (\rec)} 
\label{alg:recursive} 

\begin{algorithmic}[1]

\setcounter{ALC@line}{5}

\FOR{\textbf{each} $S$ in $\{S_1, S_2, ... S_k\}$ - $S_{seed}$}
\STATE $L_p^r \gets \emptyset$
\STATE  $L_p^r \gets$ RadRecursive($r_r$, $\alpha$, $p$, $S$, $L_p^r$) 
\STATE  $L_S\gets L_S \cup L_p^r$
\ENDFOR

\end{algorithmic}

\end{algorithm}

\end{subalgorithms}

\begin{subalgorithms}
\begin{algorithm} [htp]
\caption{RadRecursive} 
\label{proc:recursive}

\begin{algorithmic}[1]

\INPUT  $r_r$, $\alpha$, $p$, $S$, $L_p^r$
\OUTPUT  $L_p^r$
\STATE $R \gets API(p, r_r, S)$ \Comment{Query $S$ with $r_r$} 
\STATE $L_p^r \gets L_p^r \bigcup R$
\IF{$\lvert R \rvert < M_S$ } 

\STATE \textbf{return} $L_p^r$ \Comment{The area is not dense}
\ELSE 
\STATE RadRecursive($\frac{r_r}{\alpha}$, $\alpha$, $p$, $S$, $L_p^r$) \Comment{Call with new $r_r$}
\ENDIF

\end{algorithmic}

\end{algorithm}

Alg. \ref{alg:recursive} (\rec) has a modification in line 7 of Alg. \ref{alg:seed}, where a new algorithm is called. Instead of querying with a static $r$, we perform a recursive procedure to adjust the radius. Alg. \ref{proc:recursive} takes the following parameters: the radius $r_r$, the coefficient $\alpha$ which is used to reduce the radius, the point $p$ that comes from the seed and the queried source $S$. The stopping condition of the recursive procedure is retrieving less than $M_S$ locations (line 4).

The extraction of locations is a  dynamic procedure on live data so it is not possible to detect beforehand which of the methods will perform better.  \rec\ promises to be adaptive and deal with dynamic situations. However, the number of requests are limited and our \rec\ performs extra requests to adjust the radius. Moreover, we have no control over the number of requests needed by \rec\ but under the following assumption, we know that \rec\ converges:

\begin{assumption}
Let $S$ be a source and $L(S)$ its locations. For each point $p$ there exists a radius $r_p$ such that the surface covered by $Circle (p,r_p)$ contains less than $M_S$ locations.
\label{assumption:lessM}
\end{assumption}

The value of $r_p$ varies depending on the density of the region. If the point $p$ is in a sparse area, the value of $r_p$ is large and vice versa. We assume that there will always be an $r_p$ such that $\vert L_p^{r_p} \vert < M_S$. \rec\ performs several requests decreasing $r$ by $\alpha$ until $r_p$ is found and  \rec\ reaches the stopping condition. Given Assumption \ref{assumption:lessM}, we guarantee that \rec\ performs a finite number of requests.

\subsection{Optimizing the Point Selection}
\label{sec:opt}

\begin{figure}[b]
 \centering  \includegraphics[width=0.18\textwidth]{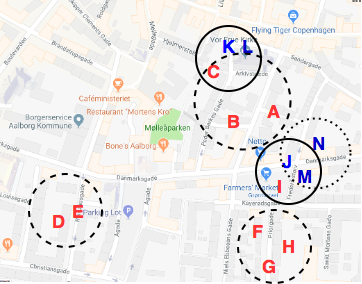} \includegraphics[width=0.12\textwidth]{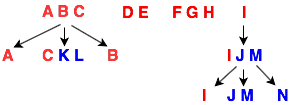}
    \caption{\opt}
    \label{fig:ideaoptimized}
\end{figure}
All \seed\ algorithms are based on exhaustive querying. However, some seed points might be quite close to each other, resulting in redundant API requests. A naive solution is clustering points together if they are within a threshold distance. This solution has the following drawbacks: (i) it is sensitive to the parameters of the clustering, (ii) it uses the density of the seed but does not adapt to the source data distribution. We propose \opt\ that minimizes the number of requests and overcomes these challenges. The idea behind \opt\ is to cluster the seed points using DBSCAN  \cite{ester1996density} (the best clustering algorithm for spatial data with noise) and to query with the centroids of the clusters. If the results size is maximal, then there is a possibility that this is a dense area. Afterwards, we apply DBSCAN on \emph{the union of the points of the current cluster and the points retrieved from the API request}. Depending on the data distribution of the source, we move the focus to the dense areas that we discover. This procedure continues recursively until the result size of the request is less than the maximal (based on Theorem \ref{theorem:stoprecursive}).

Fig. \ref{fig:ideaoptimized} shows a simple example of \opt. The red points come from the seed whereas the blue ones are in the source. The initial DBSCAN will cluster together (A, B, C), (D, E), (F,G, H) and I. After the querying with the centroids of these clusters, only clusters (A, B, C) and I will continue further. The new clusters for (A, B, C) will be A, B and (C, K, L), where K and L are points from the source. For cluster I, we query with the centroid of (I, J, M). In the third step cluster (I, J, M) is divided to I, (J,M) and N, where N is a new point discovered from the second step. \opt\ is formalized in Alg. \ref{alg:seedoptimized}. After a source is chosen, its points are clustered with DBSCAN (line 5) using $\epsilon$ as minimal distance between points and $m$ as the number of points that a cluster should have. For each of the centroids of the clusters, we call RadRecursive* (Alg. \ref{proc:DBSCANrecursive}), which is similar to its parent version, RadRecursive (Alg. \ref{proc:recursive}) regarding the stopping condition and the adaptive radius but differs from line 6 and on (the \emph{else} clause). If the area is dense, then we split the cluster by taking into consideration not only the set of points in the cluster ($C_p$) but also the retrieved points from the source ($R$). DBSCAN is run again with new parameters (line 6). For each centroid $c'$ of the result $C'$ we call the algorithm again with the adjusted parameters. Note that in the case of Twitter, the majority of results $R$ is polygons. Therefore we modify line 6 in  Alg. \ref{proc:DBSCANrecursive} with (i) the centroids of the polygons and we denote this version of \emph{MSSD*-C}  or (ii) the nearest point of the polygon to the queried point $p$ and we denote this version as \emph{MSSD*-N}.

\begin{algorithm} [!htp]

\caption{RadRecursive*} 
\label{proc:DBSCANrecursive}

\begin{algorithmic}[1]

\INPUT  $r_r$, $\alpha$, $<p, C_p>$, $S$, $\epsilon$, $m$, $L_p^r$
\OUTPUT  $L_p^r$
\setcounter{ALC@line}{5}
\STATE $\{C\}' \gets DBSCAN(C_p \bigcup R, \frac{\epsilon}{\alpha}, \frac{m}{\alpha})$ \Comment{DBSCAN on the union of $C_p$ and $R$ with new parameters}
\FOR{\textbf{each}  $<c',C'>$ }
\STATE RadRecursive*($\frac{r_r}{\alpha}$, $\alpha$, $<c',C'_c>$, $S$, $\frac{\epsilon}{\alpha}$, $\frac{m}{\alpha}, L_p^r$) 
\ENDFOR

\end{algorithmic}

\end{algorithm} 
\end{subalgorithms}

\begin{algorithm} [!htp]

\caption{\opt\ algorithm} 
\label{alg:seedoptimized} 
\begin{algorithmic}[1]
\INPUT A set of sources $\{S_1, S_2, ...S_n\}$, radius $r$
\OUTPUT    $\{L_{S_1}^*, L_{S_2}^*, ... L_{S_k}^* \}$
\FOR{\textbf{each} $S$ in $\{S_1, S_2, ...S_k\}$ - $S_{seed}$}
\STATE $L_S\gets L_I=\bigcup\limits_{i=1}^{k}{L_i}$ \Comment{Initialize each $L_S$ with $L_I$}
\ENDFOR
\STATE  Let $S_{seed}$ be the source with the most points in $\{S_1, S_2, ...S_k\}$, $L_{seed}$ its locations  and $P$ the distinct points in  $L_{seed}$

\STATE $\{C\} \gets DBSCAN(P), \epsilon, m) $

\FOR{\textbf{each} $S$}
\FOR{\textbf{each} $<c,C>$}

\STATE $L_p^r \gets \emptyset$
\STATE $L_p^r \gets $ RadRecursive*$(r, \alpha, <c, C_c>, S, \epsilon, m, L_p^r)$ 
\STATE  $L_S \gets L_S \cup L_p^r$
\ENDFOR

\ENDFOR

\RETURN{\  $\{L_{S_1}^*, L_{S_2}^*, ... L_{S_k}^* \}$}
\end{algorithmic}
\end{algorithm}

\opt\ has these advantages: (i) \opt\ manages better the requests by using the centroids of clusters rather than all the points in a cluster, (ii) \opt\ is not sensitive to parameters because it uses an adaptive algorithm to learn them for each of the sources, and (iii) while querying, \opt\ adapts the center of the circle depending on the locations found by the previous query. Let us now suppose that the optimal combination of pairs of $<p^*, r^*>$ that retrieve the maximal $L^*$ exists. In order to compare our solution to the optimal, let us first prove the submodularity of our problem.


\begin{theorem}
Optimizing the data extraction  based on API calls is a monotone submodular problem.
\label{property:submodular}
\end{theorem}

\begin{proof}
An API call takes $<p,r>$ as parameters and retrieves $L_p^r$ locations. Let us denote as $\gamma (p,r)$ the \emph{gain} (\emph{new locations}) that $API(p,r)$ brings. Note that an extra API call is effective as long as it contributes to the union of the results of the previous calls. To prove the \emph{submodularity}, we need to show that $\gamma (P' \cup p,r) \geq \gamma (P \cup p,r)$ if $P' \subset P$.

Let us consider a set of points $P$ and $P'$ such that  $P' \subset P$. The locations retrieved from $P'$ are $\bigcup_{i=1}^{\vert P' \vert} L_{p_i}^r$ and the locations retrieved from $P$ are $\bigcup_{i=1}^{\vert P \vert} L_{p_i}^r$.  Since  $P' \subset P$, $\bigcup_{i=1}^{\vert P' \vert} L_{p_i}^r \subseteq \bigcup_{i=1}^{\vert P \vert} L_{p_i}^r$.  
Let us consider a new point $p$. and $L_p^r$ the result of $API(p,r)$. Since  $\bigcup_{i=1}^{\vert P' \vert} L_{p_i}^r \subseteq \bigcup_{i=1}^{\vert P \vert} L_{p_i}^r$, then $(L_p^r \cap  (\bigcup_{i=1}^{\vert P' \vert} L_{p_i}^r))  \subseteq  (L_p^r \cap  (\bigcup_{i=1}^{\vert P \vert} L_{p_i}^r))$. As a result, $\gamma (P' \cup p, r) \geq \gamma (P \cup p, r)$. In order to prove the \emph{monotonicity}, for every $P' \subseteq P $,  $\vert \bigcup_{i=1}^{\vert P' \vert} L_{p_i}^r \vert \leq \vert \bigcup_{i=1}^{\vert P \vert} L_{p_i}^r \vert$. So, the more we increase the set of seed points, the more locations we get. It is simple to show that $\bigcup_{i=1}^{\vert P \vert} L_{p_i}^r = (\bigcup_{i=1}^{\vert P' \vert} L_{p_i}^r) \cup (\bigcup_{i-i}^{\vert P-P'\vert}L_{p_i}^r) $  so $(\bigcup_{i=1}^{\vert P' \vert} L_{p_i}^r) \cup (\bigcup_{i-i}^{\vert P-P'\vert}L_{p_i}^r) \supseteq \bigcup_{i=1}^{\vert P' \vert} L_{p_i}^r$. Hence, $\vert \bigcup_{i=1}^{\vert P' \vert} L_{p_i}^r \vert \leq \vert \bigcup_{i=1}^{\vert P \vert}$. \end{proof}

Our \opt\ tries to solve the data extraction problem by providing a solution that starts with initial centroids and then splits further if the area looks promising in terms of density. However, we extract only $M_S$ locations in one call and this sample might not be representative if the amount of the actual locations in the area may be quite large. So, if the sample of the $M_S$ points misses some dense areas, our DBSCAN will classify those as outliers and we will not query further. Thus, our solution is \emph{greedy} because it makes a locally optimal solution regarding which API calls to perform in step $i+1$ based on the information of step $i$.

\begin{theorem}
The greedy solution \opt\ of our monotone submodular problem performs at least $1 - \frac{1}{e}$ as good as the optimal solution in terms of maximizing the number of locations, where $e$ is the base of the natural logarithm.
\label{corollary:greedy}
\end{theorem}

A greedy approach to a monotone submodular problem generally offers a fast and simple solution that is guaranteed to be at least $1 - \frac{1}{e}$ as good as the optimal solution \cite{nemhauser1978analysis}. The proof uses the submodularity and the monotonicity to show the ratio between the greedy and the optimal solution.

\begin{proof}
Let $L^*$ be the result of the optimal solution from points $P^*$. and $L_k$ the greedy solution provided by \opt\ for $n$ requests. Due to the monotonicity, we can write:
{\scriptsize
\begin{align*}
\bigcup_{i=1}^{\vert P^* \vert} L_{p_i}^r \leq \bigcup_{i=1}^{\vert P^* \cup P' \vert} L_{p_i}^r = \bigcup_{i=1}^{\vert P' \vert} L_{p_i}^r + \sum_{j=1}^n \gamma(p_j,r) \\ \leq \bigcup_{i=1}^{\vert P' \vert} L_{p_i}^r + n ( \bigcup_{i=1}^{\vert P'+1 \vert} L_{p_i}^r - \bigcup_{i=1}^{\vert P' \vert} L_{p_i}^r)
\end{align*}
}
{\scriptsize
\begin{align*}
\bigcup_{i=1}^{\vert P^* \vert} L_{p_i}^r - \bigcup_{i=1}^{\vert P' \vert} L_{p_i}^r \leq n ( \bigcup_{i=1}^{\vert P'+1 \vert} L_{p_i}^r - \bigcup_{i=1}^{\vert P' \vert} L_{p_i}^r)
\end{align*}
}

We can then rearrange the previous equation as:
{\scriptsize
\begin{align*}
\bigcup_{i=1}^{\vert P^* \vert} L_{p_i}^r - \bigcup_{i=1}^{\vert P' \vert} L_{p_i}^r \leq n ((  \bigcup_{i=1}^{\vert P^* \vert} L_{p_i}^r - \bigcup_{i=1}^{\vert P' \vert} L_{p_i}^r) - (\bigcup_{i=1}^{\vert P^* \vert} L_{p_i}^r - \bigcup_{i=1}^{\vert P'+1 \vert} L_{p_i}^r)) \end{align*}
}
and we use  $\delta_i$ to represent $\bigcup_{i=1}^{\vert P^* \vert} L_{p_i}^r - \bigcup_{i=1}^{\vert P' \vert} L_{p_i}^r$ so we can rewrite:
$\delta_{i} \leq n (\delta_i - \delta_{i+1})$ and finally $\delta_{i+1} \leq (1-\frac{1}{k}) \delta_i$. So, for every $k \leq n$ we can write  $\delta_{k} \leq (1-\frac{1}{n})^k \delta_0$. Note that $\delta_0 = \bigcup_{i=1}^{\vert P^* \vert} L_{p_i}^r - \bigcup_{i=1}^{\vert \emptyset \vert} L_{p_i}^r = \bigcup_{i=1}^{\vert P^* \vert} L_{p_i}^r$. Moreover, for all  $x\in R$, $1-x \leq e^{-x}$. So finally, we can write that $\delta_{k} \leq (1-\frac{1}{n})^k \bigcup_{i=1}^{\vert P^* \vert} L_{p_i}^r \leq e^{-\frac{k}{n} \bigcup_{i=1}^{\vert P^* \vert} L_{p_i}^r}$. By substituting $\delta_k$ with  $\bigcup_{i=1}^{\vert P^* \vert} L_{p_i}^r - \bigcup_{i=1}^{\vert P^k \vert} L_{p_i}^r $, rearranging and finally replacing $\bigcup_{i=1}^{\vert P^k \vert} L_{p_i}^r $ with its result $L_k$ and $\bigcup_{i=1}^{\vert P^* \vert} L_{p_i}^r $ with $L^*$, we have:
$L_k \geq (1-e^{-\frac{k}{n}})L^*$ and for $l=k$ (lower bound) we have:
$L_k \geq (1- \frac{1}{e})L^*$.
\end{proof}




\section{Experiments}

In this section, we test our approach on the sources presented in Section \ref{sec:limitations}  and compare with the existing baselines.

\subsection{ MSSD Experiments}
\label{sec:experiments}
We run \seed\ algorithms using Krak as the seed source as it is the richest in terms of locations, points, and categories.
We compare the results of the baseline (the initial locations of sources \un) to  \fix\ which uses a fixed radius of 2 km (Alg. \ref{alg:seed}), \den\ with a density-based approach to define the radius (Alg. \ref{alg:densitybased}), \nn\     with a nearest-neighbor method to define a flexible radius (Alg. \ref{alg:nearestneighbor}) and \rec\ with a recursive method that starts with a radius of 16 km (the largest values accepted by all sources) and reduces the radius by a coefficient $\alpha$ = 2 (Alg. \ref{alg:recursive}). Some APIs allow only an integer radius in the granularity of km so $\alpha$ = 2 is the smallest integer value accepted. Fig. \ref{fig:barsseeddriven} illustrates the improvement in the extracted data volume from each source by each version of \general\ over \un. Krak is not included since it is the \emph{seed}.
Google Places (GP) has the highest improvement of 98.4 times more locations extracted by \rec. Flickr had 4,084 locations initially, which become 4.3 times more with \fix\ and above 9.5 times more with \den, \nn\ and \rec. In Foursquare (FSQ) and Yelp, \fix\ extracts 3 and 2 times more locations respectively but \den, \nn, and \rec\ retrieve up to 3.5. Twitter returns 10.7 times more with \rec\ but still has a low number of locations overall.  These values show very good results regarding the assumption that \emph{in spite of their different scopes, all the sources relate to the same physical world}. \rec\ performs the best with an improvement of 14.3 times more than \un\ and outperforms all other versions because the radius is recursively adapted to the source.

  \begin{figure}[htp]
    \centering
 \includegraphics[width=0.65\linewidth]{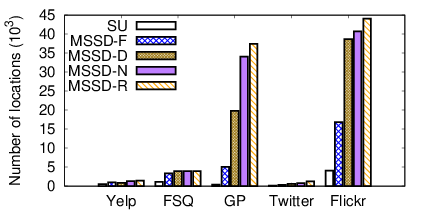}
   \caption{Number of locations extracted per source}
    \label{fig:barsseeddriven}

\end{figure}


\begin{figure}[htb]
\begin{subfigure}[b]{0.5\linewidth}
    \centering
    \includegraphics[width=\linewidth]{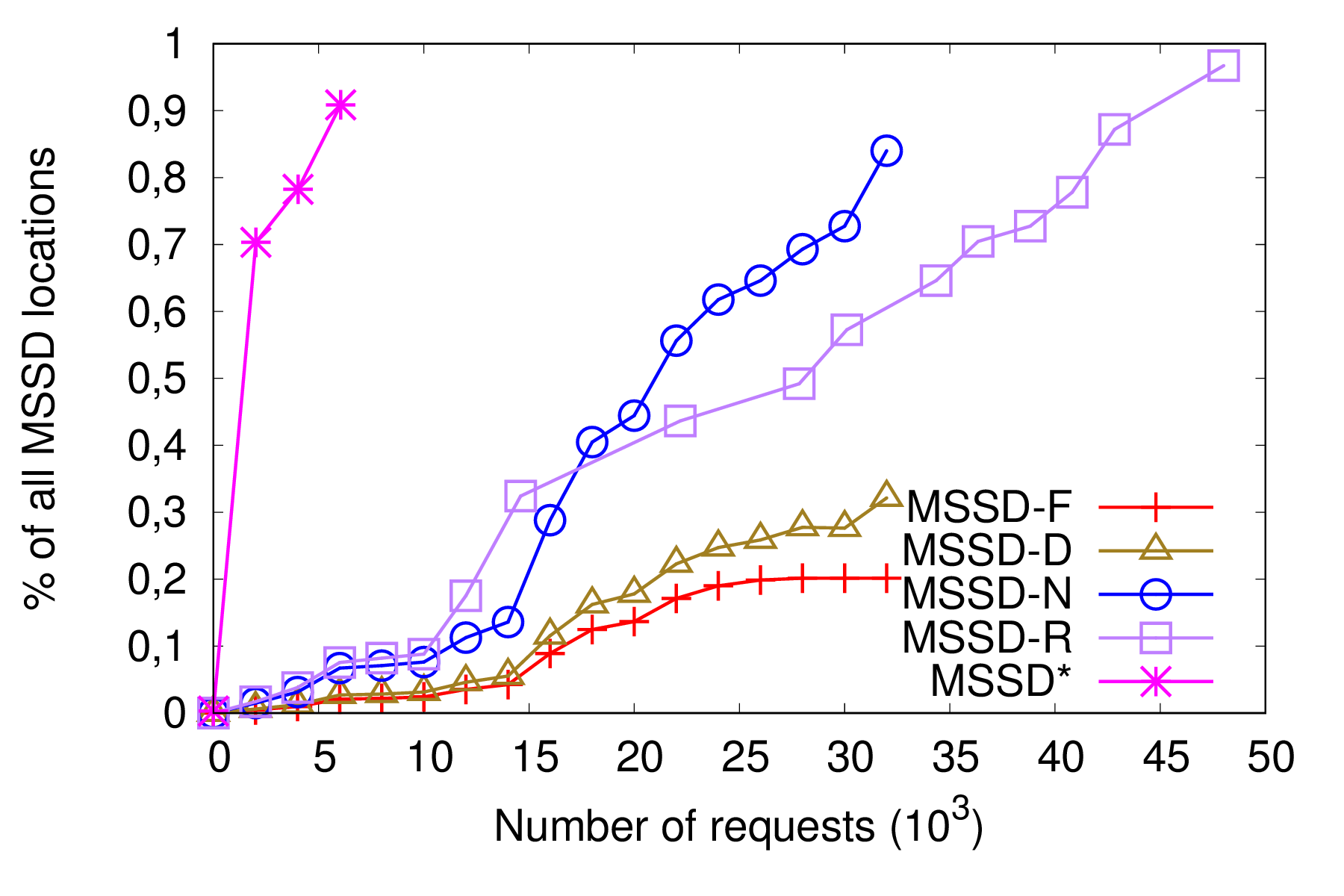}
    \caption{Yelp} 
    \label{fig:yelpopt}
    
  \end{subfigure}
  \begin{subfigure}[b]{0.5\linewidth}
    \centering
    \includegraphics[width=\linewidth]{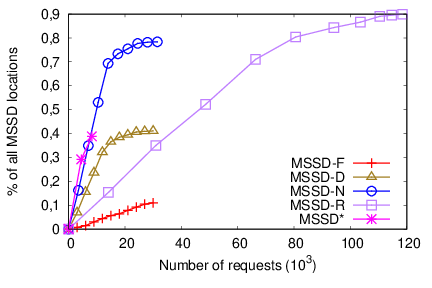}
    \caption{Google Places} 
   \label{fig:gpopt}
   
  \end{subfigure} 
  
  \begin{subfigure}[b]{0.5\linewidth}
    \centering
    \includegraphics[width=\linewidth]{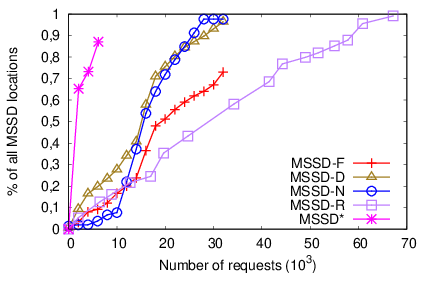}
      \caption{Foursquare} 
   \label{fig:fsqopt}
   
  \end{subfigure}\begin{subfigure}[b]{0.5\linewidth}
    \centering
    \includegraphics[width=\linewidth]{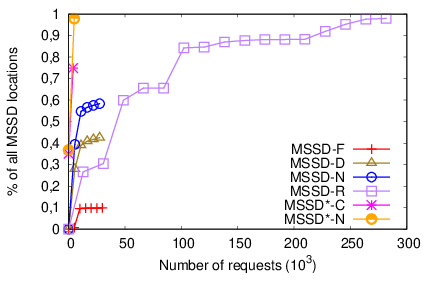}
      \caption{Twitter} 
   \label{fig:twitteropt}
   
  \end{subfigure} 
  \begin{subfigure}[b]{\linewidth}
    \centering \includegraphics[width=0.5\linewidth]{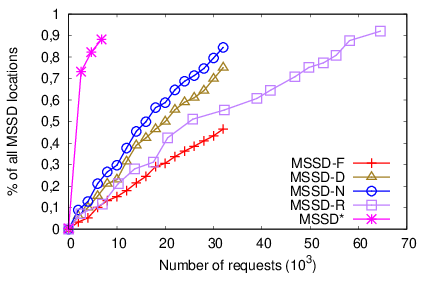}
      \caption{Flickr} 
   \label{fig:tflickropt}
   
  \end{subfigure} 

  \caption{Requests versus locations for different \general\ algorithms with Krak as seed} 
  \label{fig:optversions}

\end{figure}

\begin{table*}[htb]
\centering
\scriptsize
\begin{tabular}{@{}|ll|llllllll|lllll|@{}}
\toprule
\multicolumn{1}{|l|}{\multirow{2}{*}{\textbf{Sources}}}   & \multirow{2}{*}{\textbf{\begin{tabular}[c]{@{}l@{}}Req\\ vs loc\end{tabular}}} & \multicolumn{8}{c|}{\textbf{Alpha}}                                   & \multicolumn{5}{c|}{\textbf{Radius}}       \\ \cmidrule(l){3-15} 
\multicolumn{1}{|l|}{}                                    &                                                                                & 2      & 4      & 6      & 8      & 10     & 12     & 14     & 16     & 1      & 4      & 8      & 12     & 16     \\ \midrule
\multicolumn{1}{|l|}{\multirow{2}{*}{\textbf{FSQ}}}       & \% req                                                                         & 13.2\% & 13.0\% & 12.9\% & 12.9\% & 12.9\% & 12.9\% & 12.9\% & 12.9\% & 12.8\% & 12.9\% & 13.1\% & 13.3\% & 13.2\% \\
\multicolumn{1}{|l|}{}                                    & \% loc                                                                         & 88.3\% & 78.1\% & 75.4\% & 74.8\% & 73.0\% & 71.5\% & 70.3\% & 69.5\% & 43.2\% & 82.9\% & 86.4\% & 88.1\% & 88.3\% \\ \midrule
\multicolumn{1}{|l|}{\multirow{2}{*}{\textbf{Flickr}}}    & \% req                                                                         & 15.8\% & 15.4\% & 15.2\% & 15.1\% & 15.0\% & 15.0\% & 15.0\% & 15.0\% & 13.2\% & 14.4\% & 15.3\% & 15.7\% & 15.8\% \\
\multicolumn{1}{|l|}{}                                    & \% loc                                                                         & 96.5\% & 91.9\% & 86.9\% & 85.1\% & 82.9\% & 81.4\% & 80.3\% & 79.0\% & 49.1\% & 88.5\% & 94.8\% & 95.8\% & 96.5\% \\ \midrule
\multicolumn{1}{|l|}{\multirow{2}{*}{\textbf{GP}}}        & \% req                                                                         & 9.3\%  & 9.0\%  & 8.8\%  & 8.7\%  & 8.6\%  & 8.5\%  & 8.5\%  & 8.4\%  &   7.6\%  & 9.0\%  & 9.3\%  & 9.3\%  & 9.3\%   \\
\multicolumn{1}{|l|}{}                                    & \% loc                                                                         &38.4\% & 34.3\% & 32.9\% & 32.4\% & 31.6\% & 30.9\% & 30.4\% & 30.1\% & 33.6\% & 37.2\% & 37.3\% & 38.2\% & 38.4\% \\ \midrule
\multicolumn{1}{|l|}{\multirow{2}{*}{\textbf{Yelp}}}      & \% req                                                                         & 17.5\% & 17.4\% & 17.4\% & 17.4\% & 17.4\% & 17.4\% & 17.4\% & 17.4\% & 17.3\% & 17.4\% & 17.4\% & 17.4\% & 17.5\% \\
\multicolumn{1}{|l|}{}                                    & \% loc                                                                         & 98.2\% & 95.8\% & 93.7\% & 95.1\% & 93.4\% & 90.6\% & 89.8\% & 89.7\% & 51.1\% & 94.2\%	& 97.7\% & 98.3\% & 98.2\% \\ \midrule
\multicolumn{1}{|l|}{\multirow{2}{*}{\textbf{Twitter-C}}} & \% req                                                                         & 3.0\%  & 3.0\%  & 3.0\%  & 3.0\%  & 3.0\%  & 3.0\%  & 3.0\%  & 3.0\%  & 3.0\%  & 3.0\%  & 3.0\%  & 3.0\%  & 3.0\%  \\
\multicolumn{1}{|l|}{}                                    & \% loc                                                                         & 76.8\% & 58.6\% & 58.1\% & 57.1\% & 55.7\% & 52.8\% & 52.3\% & 52.3\% & 53.4\% & 61.5\% & 60.0\% & 58.9\% & 76.8\% \\ \midrule
\multicolumn{1}{|l|}{\multirow{2}{*}{\textbf{Twitter-N}}} & \% req                                                                         & 3.0\%  & 3.0\%  & 3.0\%  & 3.0\%  & 3.0\%  & 3.0\%  & 3.0\%  & 3.0\%  & 3.0\%  & 3.0\%  & 3.0\%  & 3.0\%  & 3.0\%  \\
\multicolumn{1}{|l|}{}                                    & \% loc                                                                         & 99.5\% & 99.4\% & 99.0\% & 98.6\% & 98.3\% & 98.2\% & 97.9\% & 97.8\% & 86.4\% & 97.0\% & 97.1\% & 98.4\% & 99.5\% \\ \bottomrule
\end{tabular}
\caption{\% of requests versus \% of locations for \opt\ relative to \rec\ for different values of $\alpha$ and $r$}
\label{tab:alphatable}

\end{table*}

\begin{figure*}[htb]
\begin{subfigure}[b]{0.2\linewidth}
    \centering
    \includegraphics[width=\linewidth]{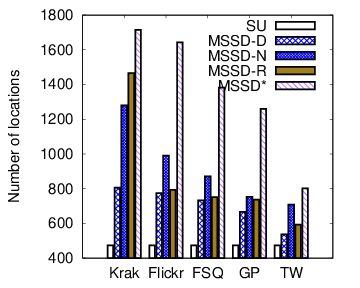}
    \caption{Yelp} 
    \label{fig:yelpM}
  \end{subfigure}
  \begin{subfigure}[b]{0.2\linewidth}
    \centering
    \includegraphics[width=\linewidth]{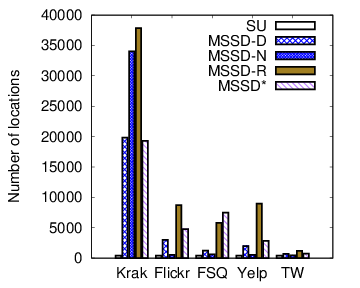}
    \caption{GooglePlaces} 
   \label{fig:gpM}
  \end{subfigure} 
  \begin{subfigure}[b]{0.2\linewidth}
    \centering
    \includegraphics[width=\linewidth]{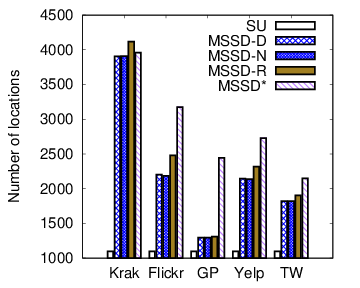}
    \caption{Foursquare} 
   \label{fig:fsqM}
  \end{subfigure}\begin{subfigure}[b]{0.2\linewidth}
    \centering
    \includegraphics[width=\linewidth]{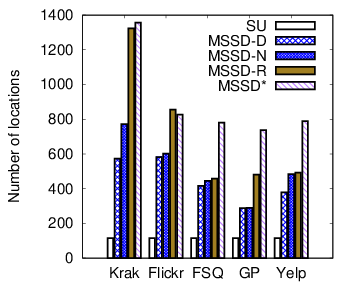}
    \caption{Twitter} 
   \label{fig:twM}
  \end{subfigure}\begin{subfigure}[b]{0.2\linewidth}
    \centering
    \includegraphics[width=\linewidth]{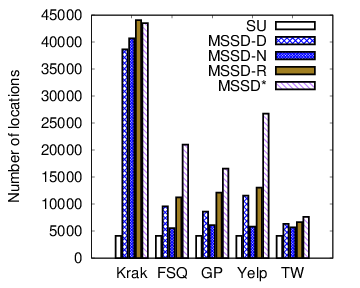}
    \caption{Flickr} 
   \label{fig:flickrM}
  \end{subfigure}
  
  \caption{ \general\ results with different seeds for all sources} 
  \label{fig:severalseed}

\end{figure*}

\rec\ performs the best compared to all \seed\ versions but it uses extra requests until it fixes the radius. The extra requests of \rec\ vary from 1.5-2 times more for Yelp, Foursquare, and Flickr to 8 times more in Twitter and Google Places.  We ran the optimized version  \opt\ (Alg. \ref{alg:seedoptimized}) for each of the sources with initial radius of 16 km and initial $m$ = 10 and $\epsilon$ = 500 meters as parameters of DBSCAN.  $m$, $\epsilon$ and $r$ are recursively reduced by $\alpha$ = 2.  We compared \fix, \den, \nn, \rec\ and \opt\ regarding the number of requests performed and the locations retrieved. The results for each source are presented in Fig. \ref{fig:optversions}. The number of requests is in the x-axis whereas the number of locations  is expressed as the percentage of the total of distinct locations extracted by all methods. \rec\ provides the highest percentage of locations (above 95\%) for all the sources but considerably more requests. For instance,  for Google Places  (Fig. \ref{fig:gpopt}) and for Twitter (Fig. \ref{fig:twitteropt}), \rec\ need respectively 3.8 and 8.7 times more requests than the \general\ versions with fixed number of requests (\fix, \den, \nn). For the same number of requests, \nn\ provides a higher percentage of locations compared to \fix\ and \den\ for all the sources. \opt\ is the most efficient in terms of requests. For all sources except Google Places, \opt\ gets around 90\% of the locations with around 25\% of the requests of \fix, \den\, and \nn. With regard to \rec, \opt\ uses 12\%-15\% of \rec\ requests for Flickr, Yelp and Foursquare, 8.5\% of \rec\ requests for Google Places and 2.7\% of \rec\ requests for Twitter. In Google Places, \opt\ is able to retrieve only 40\% of the locations. This result is explained by the small result size of Google (Table \ref{tab:categories}); despite the fact that \opt\ uses the requests wisely in dense areas, a request can retrieve at most 20 locations. \nn\ extracts 2 times more locations than \opt\ but with 3 times more requests. As previously mentioned in Section \ref{sec:opt}, we propose two versions of \opt\ in the case of Twitter: a centroid-based (\emph{MSSD*-C}) and a nearest point method (\emph{MSSD*-N}). \emph{MSSD*-C} retrieves 20\% more locations than \emph{MSSD*-N} using the same number of requests. To conclude, \opt\ guarantees the best compromise for all the sources. It is able to quickly reach a high percentage of locations with remarkably fewer requests.

\textbf{Setting $\alpha$ and radius $r$}

So far, \opt\ used values of $\alpha$ (the smallest integer to reduce the radius) and $r$ (the largest radius accepted by all sources) that guarantee a slow and relatively full exploration of the area in exchange for more requests. In this section, we test different values of $\alpha$ and $r$ for \opt. When $\alpha$ is bigger or $r$ is smaller, fewer requests are performed, some areas are missed, and consequently, fewer locations are retrieved. Table \ref{tab:alphatable} provides the trade-offs in terms of percentage of requests and percentage of locations of \opt\ with regards to \rec\ for each $\alpha$ (while fixing the radius at 16 km) and for each $r$ (while fixing  $\alpha$ at 2)(Section B.2 in Appendix for more details). In all the cases, it is obvious that the extra requests of \opt\ with small values of $\alpha$ are rewarded with a higher percentage of locations. For example, for 0.3\% more requests, we retrieve 18.8\% more locations in Foursquare. In Flickr, for 0.8\% more requests, we retrieve 17.5\% more locations. Similarly, starting with a big radius is safer and more rewarding. For instance, Foursquare and Yelp perform less than 0.4\% of the requests to get around 46\% more locations when starting with $r$=16km compared to $r$=1km.  \emph{A risk-averted selection of parameters turns out to provide a good trade-off between the number of requests and number of locations because \opt\ adapts to the density of the region and still manages the requests carefully}. Thus, the algorithm is robust to different parameter settings and fine-tuning is not needed.

\begin{figure*}[htb]
\begin{subfigure}[b]{0.2\linewidth}
    \centering
    \includegraphics[width=\linewidth]{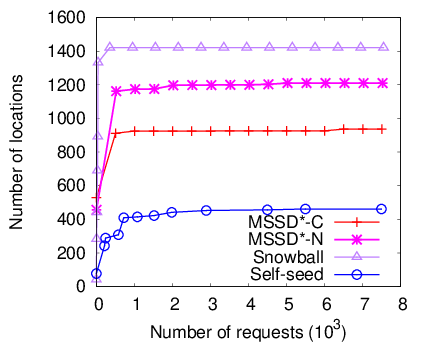}
    \caption{Twitter} 
    \label{fig:snowballtwitter}
  \end{subfigure}
  \begin{subfigure}[b]{0.2\linewidth}
    \centering
    \includegraphics[width=\linewidth]{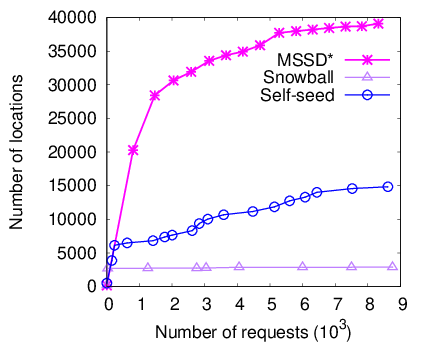}
    \caption{Flickr} 
   \label{fig:snowballFlickr}
  \end{subfigure} 
  \begin{subfigure}[b]{0.2\linewidth}
    \centering
    \includegraphics[width=\linewidth]{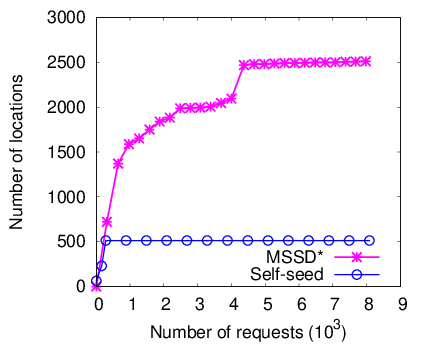}
    \caption{Foursquare} 
   \label{fig:baselinesFSQ}
  \end{subfigure}\begin{subfigure}[b]{0.2\linewidth}
    \centering
    \includegraphics[width=\linewidth]{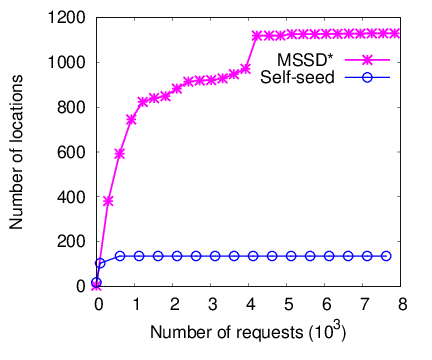}
    \caption{Yelp} 
   \label{fig:baselinesYelp}
  \end{subfigure}\begin{subfigure}[b]{0.2\linewidth}
    \centering
    \includegraphics[width=\linewidth]{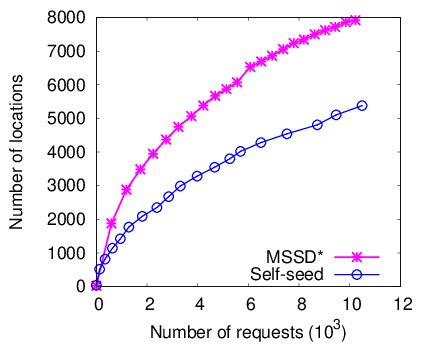}
    \caption{Google Places} 
   \label{fig:baselinesGoogle}
  \end{subfigure}
  \caption{Number of request versus number of locations for \opt, \emph{Snowball} and \emph{Self-seed}} 
  \label{fig:snowball}

\end{figure*}

\textbf{Choosing a different seed}. Previously, we chose Krak as the seed due to its richness. However, Krak is limited to some countries. In order to show that our \seed\ algorithms are applicable to any type of seed (preferable a rich source), we ran \den, \nn, \rec\ and \opt using as seed Flickr, Foursquare, Yelp, Google Places, and Twitter. The seed points are those discovered by the initial querying in Section \ref{sec:datascarcity}. The results are presented in Fig. \ref{fig:severalseed}. Even though Krak performs the best, the other sources which provide significantly fewer seed points (see Table \ref{tab:nrpoints}) are able to achieve comparable results. Recall that Krak has 14 times more seed points than Flickr, 30 times more than Foursquare, 70 times more than Yelp, 91 times more than Google Places and 295 times more than Twitter. Apart from Krak, \opt\ performs the best for Flickr, Yelp, and Foursquare. For Flickr, \opt\ with Yelp seed points retrieves 6.5 times more points than \un. For FSQ, \opt\ with Flickr and Yelp seed points retrieves 2.9 and 2.5 times more points than \un, respectively. For Yelp, \opt\ with Flickr seed points retrieves 3.5 times more data than \un, whereas Krak retrieves 3.6 times more while having 70 times more seed points. For Twitter and Google Places, \rec\ performs the best; 7.4 times more locations than \un\ with Flickr seed points in Twitter and 23.6 times more locations with Yelp seed points in Google Places. The performance of each algorithm in terms of the number of requests versus the number of locations can be found in Section B.3 in Appendix. An interesting observation is that \opt\ sometimes performs better than \rec. When we experimented with Krak, given its variety of seed points, the need to change the center of the queried area in the next call was not needed. However, \opt\ adapts very well even when given  a seed that is not rich. In the next recursive call, \opt\ uses the result of the previous call and the seed points in the area to select the next centers of the recursive calls (see Alg. \ref{alg:seedoptimized} for more details). \emph{Even when the seed source is not rich, \opt\ manages to achieve good results due to its ability to adapt the next call according to the distribution of the source.}

\textbf{Elapsed time of the experiments}. All \general\ versions are executed on live data sources, so the elapsed time is more important than the CPU time because of the \emph{bandwidth limitations} of the APIs. For instance, if the limit of a source is reached, then the thread has to sleep until the limit is reset. Our experiments were run simultaneously for all the sources in order to avoid the temporal bias, therefore all bandwidth limitations were respected at the same time. Subsequently, the elapsed time for all the sources is maximal, specifically around 1 week for the \fix, \den\ and \nn, 2 weeks for \rec\ and 1.7 days for \opt. If the algorithms are run independently, it takes on average 1 day per source for \fix, \den\, and \nn, 2 days for \rec\ and less than 6 hours for \opt. 

\subsection{Comparison with Existing Baselines}

The technique using \emph{linked accounts} \cite{hristova2016measuring,preoctiuc2013mining,armenatzoglou2013general} requires users that have declared their account in another social network, who are rare to find. From our initial querying of the sources, there were only 0.27 \% of users on Flickr with linked accounts to Twitter and  0.003 \% of users on Twitter with linked accounts to Foursquare. Since the percentage of linked accounts is unfeasibly low, a comparison with this technique makes little sense. The \emph{keyword-based} querying shows little applicability in location-based data retrieval. We conducted a small experiment using the names of cities and towns in North Denmark as keywords.
For Flickr and Twitter the precision (\% of data that falls in the queried region) was just 31.6\% and 0.85\% respectively, while the recall was less than 5\%, relative to \opt. 
 Foursquare and Yelp offer a query by term or query by location expressed as a string. The former [query by term] does not retrieve any data when queried with a city or town name. If we express the location as a string, the precision is 93\% and 85\% for Foursquare and Yelp respectively and the recall is less than 19\%, relative to \opt. In Google Places, the data retrieved is the towns and the cities themselves. For example, if we query with the keyword "Aalborg", the API will return Aalborg city only and not any other places located in Aalborg (1 request per 1 location). Even though the precision is 100\%, the recall is only 0.07\% relative to \opt.   Thus, we compare to \emph{Snowball} and to the technique mentioned in \cite{lee2010measuring}. 
\begin{table}[htb]
\centering
\scriptsize
\begin{tabular}{@{}lllll@{}}
\toprule
        &          & Nr of locations & Nr of users & Time period \\ \midrule
Twitter & Snowball & 1421                & 35              & 2015-2018  \\
        & Self-seed  & 461                 & 101             & 2017-2018    \\
        & MSSD*-C  & 936                 & 195             & 2017-2018   \\
        & MSSD*-N  & 1237                & 231             & 2017-2018   \\ \hline
Flickr  & Snowball & 2885                & 46              & 2005-2018   \\  
        & Self-seed   & 14910               & 1007            & 2005-2018 \\
        & MSSD*    & 39427               & 1740            & 2005-2018    \\
        
        \bottomrule
\end{tabular}
\caption{\emph{Snowball}, \emph{Self-seed} and \opt\ results}
\label{tab:snowball}
\end{table}

\begin{table*}[htb]
\centering
\scriptsize
\begin{tabular}
{@{}|l|l|lll|lll|lll|@{}}
\toprule
\textbf{Source}                  & \textbf{Algorithm}  & \textbf{\begin{tabular}[c]{@{}l@{}}Real \\ 80\%\end{tabular}} & \textbf{\begin{tabular}[c]{@{}l@{}}Synthetic\\ 20\%\end{tabular}} & \textbf{Req} & \textbf{\begin{tabular}[c]{@{}l@{}}Real\\ 70\%\end{tabular}} & \textbf{\begin{tabular}[c]{@{}l@{}}Synthetic\\ 30\%\end{tabular}} & \textbf{Req} & \textbf{\begin{tabular}[c]{@{}l@{}}Real\\ 50\%\end{tabular}} & \textbf{\begin{tabular}[c]{@{}l@{}}Synthetic\\ 50\%\end{tabular}} & \textbf{Req} \\ \midrule
\multirow{2}{*}{\textbf{Flickr}} & \rec & 85.62\%                                                       & 87.16\%                                                           &              & 84.44\%                                                      & 85.70\%                                                           &              & 82.71\%                                                      & 83.70\%                                                           &              \\
                                 & \opt & 64.91\%                                                       & 66.22\%                                                           & 10.16\%      & 63.59\%                                                      & 65.01\%                                                           & 10.12\%      & 62.12\%                                                      & 64.28\%                                                           & 9.80\%       \\ \midrule
\multirow{2}{*}{\textbf{Yelp}}   & \rec & 97.17\%                                                       & 99.44\%                                                           &              & 97.10\%                                                      & 99.4\%                                                            &              & 97.10\%                                                      & 99.23\%                                                           &              \\
                                 & \opt & 72.23\%                                                       & 77.81\%                                                           & 12.74\%      & 68.89\%                                                      & 76.94\%                                                           & 12.49\%      & 66.42\%                                                      & 72.59\%                                                           & 11.78\%      \\ \midrule
\multirow{2}{*}{\textbf{FSQ}}    & \rec & 94.27\%                                                       & 96.69\%                                                           &              & 94.19\%                                                      & 95.51\%                                                           &              & 94.07\%                                                      & 95.26\%                                                           &              \\
                                 & \opt & 67,75\%                                                       & 74,57\%                                                           & 10.67\%      & 66,15\%                                                      & 70,34\%                                                           & 10.62\%      & 63,60\%                                                      & 68,65\%                                                           & 10.45\%     \\ \midrule
\multirow{2}{*}{\textbf{GP}}    & \rec & 92,42 \%                                                       & 78,76\%                                                           &              & 91,60\%                                                      & 76,15\%                                                           &              & 89,94\%                                                      & 69,33\%                                                           &              \\
                                 & \opt & 34,48\%                                                       & 35,76\%                                                           & 13,85\%      & 45,80\%                                                      & 33,59\%                                                           & 9,14\%      & 38,09\%                                                      & 33,99\%                                                           & 13,49\%  \\ \midrule
\multirow{3}{*}{\textbf{Twitter}}    & \rec & 81,16 \%                                                       & 97,70\%                                                           &              & 81,16\%                                                      & 95,66\%                                                           &              & 80,92\%                                                      & 87,23\%                                                           &              \\
                                 & \emph{MSSD*-C} & 45,30\%                                                       & 63,31\%                                                           & 2,45\%      & 44,37\%                                                      & 62,09\%                                                           & 2,44\%      & 48,88\%                                                      & 68,51\%                                                           & 2,45\%    \\ & \emph{MSSD*-N} & 69,79\%                                                       & 97,70\%                                                           & 2,53\%      & 69,11\%                                                      & 94,51\%                                                           & 2,53\%      & 60,90\%                                                      & 86,61\%                                                           & 2,53\%    \\ \bottomrule
\end{tabular}
\caption{\rec\ and \opt\ performance compared to the ground truth}
\label{table:groundtruth}
\end{table*}

 \emph{Snowball}  (\cite{gao2015content,scellato2010distance}) starts with a seed of users and then traverses their network while extracting user data. In order to compare with \emph{Snowball}, we formed the seed with the users found in Section \ref{sec:limitations}. We used the same number of requests for \emph{Snowball} and \opt.   \emph{Snowball} is based on users, and consequently, it can be applied only to Twitter and Flickr.  Foursquare could not be included since the API no longer provides the check-in data extraction from users unless the crawling user has checked in himself at the venue. The technique mentioned in \cite{lee2010measuring} (we will refer to it as \emph{Self-seed}) starts with querying a specific location to get initial points. Later, other requests are performed using the seed points of the previous step. We ran \emph{Self-seed} on all our sources for the same number of requests as \opt\ (results in Fig. \ref{fig:snowball}).

It is important to note that the  Twitter API result size is 200 tweets for user-based queries and 100 for location-based ones. Moreover, the historical access of the location-based API is 2 weeks whereas for the user-based queries it is unlimited.  As a result, \emph{Snowball} in Twitter retrieves more locations in the region than versions of \opt\ (\emph{MSSD*-C} and \emph{MSSD*-N}) and \emph{Self-seed} (Fig. \ref{fig:snowballtwitter}). In the case of Flickr, \opt\ outperforms \emph{Snowball} and \emph{Self-seed} with 14 and 3 times more locations respectively. \emph{Snowball} gets most of the data in the region in the beginning and then the improvement is quite small because \emph{when using Snowball, while we traverse the network (friends of friends and so on), there is more and more data which falls outside the region of interest}. \opt\ yields a higher number of locations compared to \emph{Self-seed} as well: 5.5 times more locations for Foursquare, 9 times more locations for Yelp and 3.5 times more locations for Google Places. Note that \emph{Self-seed} in the case of directories stops yielding new locations after approximately 500 requests.  \emph{In the case of directories, after some steps, the seed points in Self-seed stop growing, converging into a dead end}. Recall that Google Places has a result size of 20 and is denser in terms of data so it has new locations for the coming steps, avoiding thus the dead end convergence. The number of users and the time period covered are presented in Table \ref{tab:snowball}. Despite the slight advantage of \emph{Snowball} in Twitter in terms of the number of locations, the data comes only from 35 users compared to 231 for \emph{MSSD*-N} and 101 for \emph{Self-seed}. Moreover, the period of time covered by the tweets in \emph{Snowball} is 3 years compared to 1 year of \opt\ versions. Regarding Flickr, the time period of the photos is the same but the number of photos and the number of users are 1-2 orders of magnitude larger for \opt\ compared to \emph{Snowball}. \emph{Self-seed} retrieves a better variety of users and locations compared to \emph{Snowball} but still contains only half the number of locations and users of \opt.

\subsection{MSSD-R and MSSD* Result Completeness}

Obtaining all locations from the sources is infeasible given the API limitations. We thus cannot get the actual ground truth of source locations. Instead, we have to approximate it. We performed the following experiment: first, we union all the locations sets from all our algorithms (\un, \fix, \den, \nn, \rec\ and \opt) to create a dataset of real data; second, we learn the distribution $\mathcal{D}$ of the locations by dividing the area in a grid of 1km x 1km and assigning each grid cell $d$ a probability $p_d \sim \mathcal{D}_d$; third, we generate synthetic locations in the area and  assign them to a grid cell $d$ with the estimated probability $p_d$. We consider the synthetic and the real data as ground truth. We implemented "simulated offline" API functions for each source, respecting the maximal result size for each of them. We ran our \rec\ and \opt\ on the ground truth data for different ratios of synthetic data as in the Table \ref{table:groundtruth}. The data retrieved by \rec\ is above 94\% of the ground truth in Yelp and Foursquare and above  80\% of the ground truth in Flickr, Google Places and Twitter. \opt\ performs the best in Yelp and Twitter (\emph{MSSD*-N}) with above 70\% of the ground truth for all ratios of real versus synthetic data, followed by Foursquare and Flickr with above 64\%.  What is more important, \emph{\rec\ and \opt\ are seen to be robust regardless of the ratio of synthetic to real data}. Despite the fact that \opt\ retrieves less than \rec, it is important to note that this result is achieved using only around 10\% of the requests of \rec. In the case of Google Places, \opt\ gets around 40\% of the ground truth because even though \opt\ uses the requests wisely to move to dense areas, it still gets only 20 locations per request. In the case of Twitter, \emph{MSSD*-N} performs better than \emph{MSSD*-C}. When querying with $p$ and retrieving a polygon $pp$, using the nearest point of $pp$ to $p$ as input for DBSCAN clusters proved to achieve better results than using the centroid of $pp$. Obviously, the centroid of a large polygon might be quite far from the queried area so the quality of the clusters reduces.

\subsection{Discussion of Experiments}

Selecting an external seed of points showed to improve the number of locations retrieved and avoid converging into a dead end like in \emph{Self-seed}. Moreover, the attempts to adapt the radius of the search according to the search region prove to be effective in retrieving more locations. \fix, \den\ and \nn\ extract on average up to 11.1 times more data than \un\ but if we adapt the radius according to the source (\rec), we extract up to 14.3 times more locations than \un. \opt\ provides a very good trade-off between the number of requests and number of locations as \opt\ extracts up to 90\% of the data of \rec\ with less than 16\% of its requests. The economic use of the requests limits the elapsed time to be less than 2 days for all the sources.  Our comparison with the \emph{Snowball} and the \emph{Self-seed} baseline shows that our seed-driven algorithm is better in terms of extracting (i) up to 14 times more locations for all the sources, (ii) in the case of Twitter and Flickr, the activity originates from \emph{a larger base of users} (up to 6.6 times more), and (iii) in the case of directories, our \seed\ avoids converging into a dead end. In a ground truth dataset, for most of the sources, our \rec\ algorithm finds 82 \% - 99 \% of the ground truth, while \opt\ with 10\% of the requests is able to guarantee 63 \% - 73\% of the ground truth. 

\section{Conclusions and future work}

This paper was motivated by the need for an efficient algorithm that extracts recent geo-social data. We formulated the problem of data extraction as an optimization problem which aims to maximize the retrieved locations while minimizing the requests. We identified the API limitations and analyzed the scarcity of the data for six sources: Krak, Yelp, Google Places, Foursquare, Twitter and Flickr. Then, we proposed a seed-driven algorithm that uses the richer source as the \emph{seed} to feed the points as API parameters to the others. \general\ versions extracted up to 14.3 times more data than \un. We further optimized our \general\ algorithm with respect to the radius of search and the query points. Specifically, in our experiments, we proved that  \opt\ is able to retrieve 90\% of the locations with less than 16\% of the requests, outperforming \den\ and \nn. Interesting directions for future research include applying machine learning for data extraction, seed selection based on other criteria (diversity in semantics, maximal spread of points, relation to the source), data integration, and data fusion of location-based data from multiple geo-social sources.

\appendix

\section{Geo-Social Dataset Age}

The challenging data extraction of geo-social data has not discouraged the research in the field. We checked the existing related work and the dataset used. The x-axis shows the year when the article was published and the y-axis indicates the latest year in the dataset. For example, if the dataset covers years 2008-2011 for an article of 2013, then the point is (2013, 2011). The labels on the points are the references in this paper. We noticed that 50\% the existing work, mostly in recent years, has been carried in datasets of 3-8 years older than the published article.

\begin{center}
\begin{tikzpicture} 
\begin{axis}[ grid=both,scale=0.85,
    xmin=2007, xmax=2018,
    ymin=2007, ymax=2018,
    xtick={2008,2010, 2012,2014,2016,2018, 2019},
    ytick={2008,2010, 2012,2014,2016,2018, 2019},
    legend pos=north west,
    ymajorgrids=true,
    grid style=dashed, 
]
\tiny
\pgfplotsset{every tick label/.append style={font=\footnotesize}}
\addplot+[color=black,
    mark=*,nodes near coords,only marks, point meta=explicit symbolic]
table[meta=label] {
x y label
2010 2009 \cite{scellato2010distance}
2012 2011 \cite{zhang2012evaluating} 
2011 2010 \cite{cho2011friendship,  scellato2011exploiting}
2011 2011 \cite{ferrari2011extracting}
2013 2012 \cite{armenatzoglou2013general}
2013 2011 \cite{gao2013exploring,ference2013location,zhang2013igslr}
2010 2008 \cite{crandall2010inferring}
2013 2010 \cite{liu2013point,wang2013location}
2010 2010 \cite{lee2010measuring}
2015 2011 \cite{feng2015personalized}
2017 2011 \cite{saleem2017location,li2017geo}
2015 2015 \cite{weiler2015geo}
2016 2013 \cite{yin2016discovering,yao2016poi}
2012 2010 \cite{noulas2012random}
2009 2008 \cite{li2009analysis}
2014 2010 \cite{emrich2014geo}
2016 2011 \cite{hristova2016measuring,zhao2016stellar}
2016 2010 \cite{cai2016using}
2017 2013 \cite{yu2017friend}
2018 2011 \cite{saleem2018effective}
2017 2010 \cite{li2017geo,zhu2017geo}
2014 2011 \cite{li2014efficient}
2015 2013 \cite{gao2015content}
2015 2014 \cite{jurgens2015geolocation}
           };
\draw [red,dashed] (rel axis cs:0,0) -- (rel axis cs:1,1);
\end{axis} 
\end{tikzpicture}
\end{center}
 \captionof{figure}{The year of the published article (x-axis) versus the latest year in their dataset (y-axis)}

 \section{Additional experiments}

\subsection{MSSD spatial and temporal improvement}

\begin{figure}[htp]
\begin{subfigure}[b]{0.5\linewidth}
    \centering
    \includegraphics[width=\linewidth]{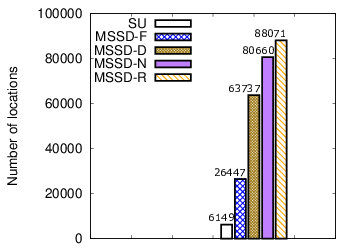}
    \caption{Spatial} 
    \label{fig:spatial}
  \end{subfigure}
  \begin{subfigure}[b]{0.5\linewidth}
    \centering
    \includegraphics[width=\linewidth]{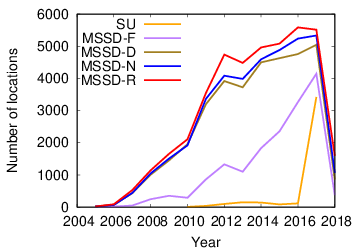}
    \caption{Temporal} 
   \label{fig:temporal}
  \end{subfigure} 
  \caption{Spatial and temporal improvement of  \seed\ } 
  \label{fig:versionsCompare}
 
\end{figure}

In this subsection, we show the spatial and temporal improvement of the \general\ versions over the baseline of \un\ (Fig. \ref{fig:versionsCompare}). We union the results retrieved from all the sources and we compare the overall improvement that each \seed\ version achieves. \fix\ has an improvement of 4.3 times more locations than \un. Despite the fact that \den\ and \nn\ have no knowledge about the source, only using the seed they retrieve 10.3 and 11.1 times more locations over \un. \nn\ is more flexible in choosing the radius, therefore it outperforms \den. \rec\ performs the best with an improvement of 14.3 times more than \un\ and outperforms all other versions of \general\ because the radius is recursively adapted to the source.  Similarly, the versions that adapt the radius outperform the fixed radius ones in terms of the temporal coverage of the data (Fig. \ref{fig:temporal}). \rec\ has the best coverage, followed by \nn\ and \den. It should be pointed out that \emph{all the seed-driven approaches contribute to getting more historical data compared to the baseline}. \rec\ performs the best compared to all \seed\ versions but it uses extra requests until it fixes the radius; 1.5-2 times more for Yelp, Foursquare, and Flickr to 8 times more in Twitter and Google Places. 

\subsection{Requests vs locations for different $\alpha$ and $r$}

In this subsection, we analyze the trade-off between the number of requests and number of locations for \opt\ in each source, for different values of $\alpha$ and $r$, having Krak as seed. The results of experimenting with $\alpha$ are presented in Fig. \ref{fig:optalpha} and with $r$ in Fig. \ref{fig:optradius}. The behavior of \opt\ with different $\alpha$ is similar in Yelp, Flickr and Twitter N but the advantage of using a smaller $\alpha$ is more obvious in  Twitter C and the last requests of Foursquare and Google Places. $\alpha$=2 guarantees a slower exploration of the area while the requests are managed carefully. As noticed by the graphs, the difference in requests is insignificant compared to the gain in the number of locations. Similarly, the bigger radius is a better option for all sources. In Yelp, Foursquare, and Flickr, the use of $r$=1km is obviously not a good choice. When we query with a small radius,  we expect that the source has almost the same distribution as the seed, which is not always the case. The seed points are used as initial points to mark an area, not to indicate exact points. On the contrary, $r$=1km is still a good option for Google Places, suggesting that it might have a similar distribution and density with Krak.

\begin{figure*}[htb]
\centering
\begin{subfigure}[b]{0.3\linewidth}
    \centering
    \includegraphics[width=\linewidth]{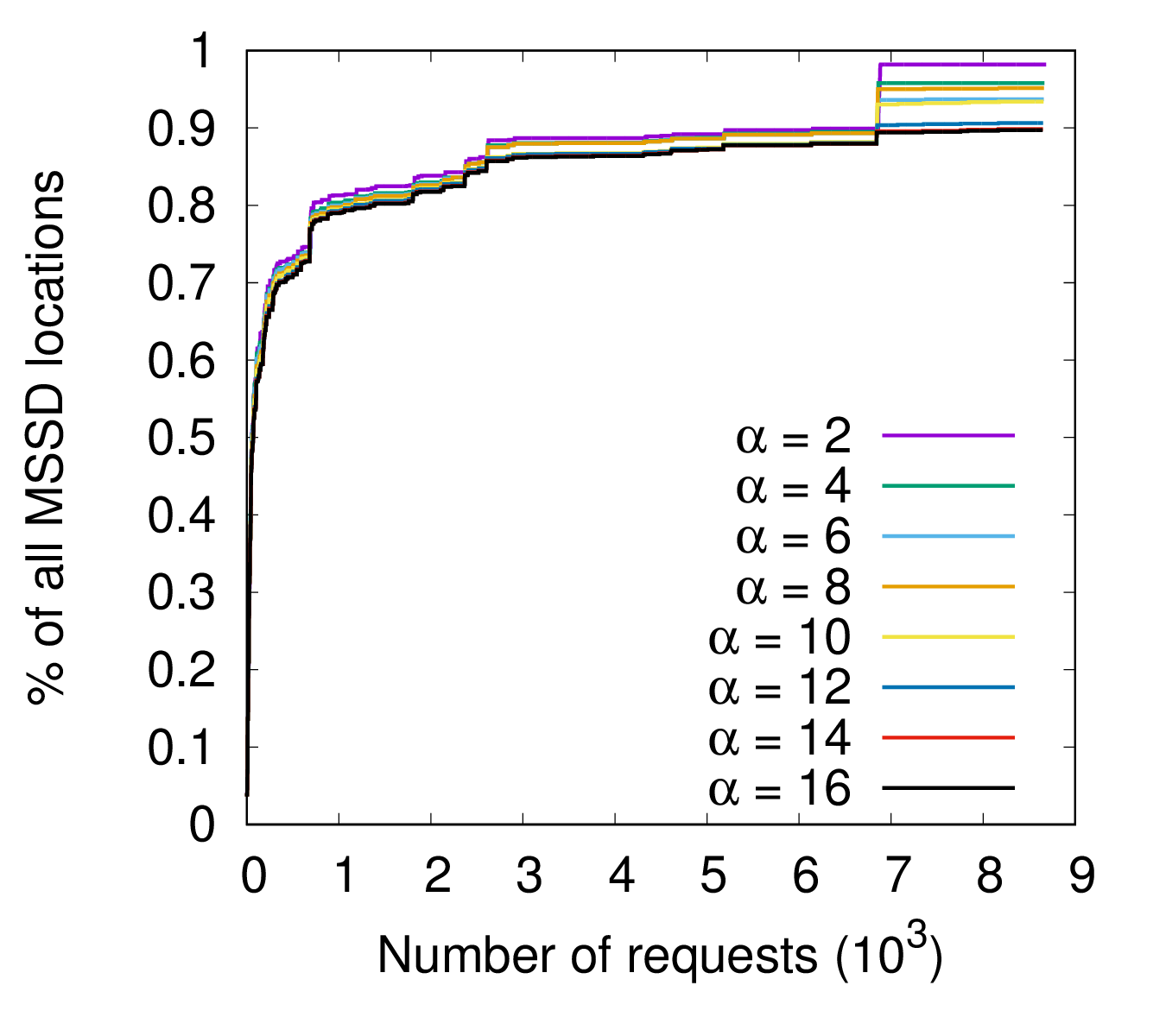}
    \caption{Yelp} 
    \label{fig:yelalpha}
    
  \end{subfigure}
  \begin{subfigure}[b]{0.3\linewidth}
    \centering
    \includegraphics[width=\linewidth]{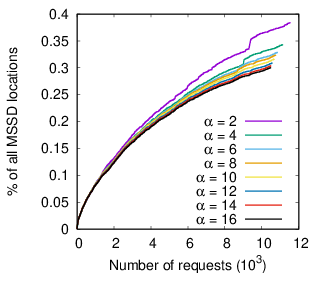}
    \caption{Google Places} 
   \label{fig:gpalpha}
   
  \end{subfigure} 
  \begin{subfigure}[b]{0.3\linewidth}
    \centering
    \includegraphics[width=\linewidth]{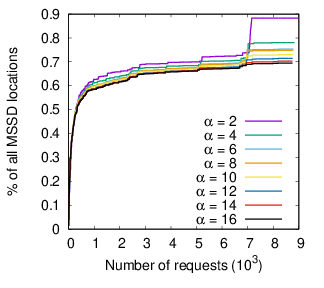}
      \caption{Foursquare} 
   \label{fig:fsqalpha}
   
  \end{subfigure}

  \begin{subfigure}[b]{0.3\linewidth}
    \centering
    \includegraphics[width=\linewidth]{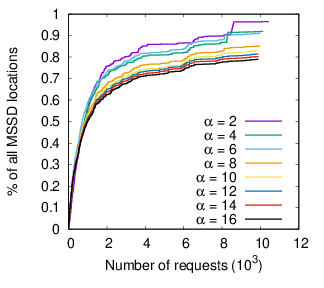}
      \caption{Flickr} 
   \label{fig:flickralpha}
   
  \end{subfigure} 
  \begin{subfigure}[b]{0.3\linewidth}
    \centering \includegraphics[width=\linewidth]{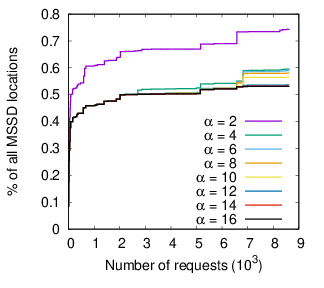}
      \caption{Twitter C} 
   \label{fig:twalpha}
   
  \end{subfigure}\begin{subfigure}[b]{0.3\linewidth}
    \centering \includegraphics[width=\linewidth]{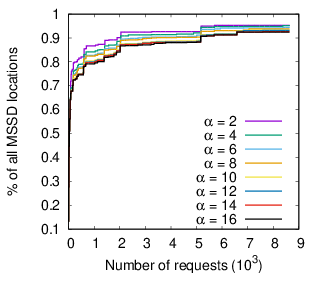}
      \caption{Twitter N} 
   \label{fig:ntwalpha}
   
  \end{subfigure} 

  \caption{Requests vs locations for different $\alpha$ in  \opt\ } 
  \label{fig:optalpha}

\end{figure*}

\begin{figure*}[htb]
\centering
\begin{subfigure}[b]{0.3\linewidth}
    \centering
    \includegraphics[width=\linewidth]{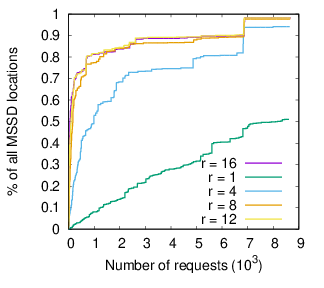}
    \caption{Yelp} 
    \label{fig:yelpradius}
    
  \end{subfigure}
  \begin{subfigure}[b]{0.3\linewidth}
    \centering
    \includegraphics[width=\linewidth]{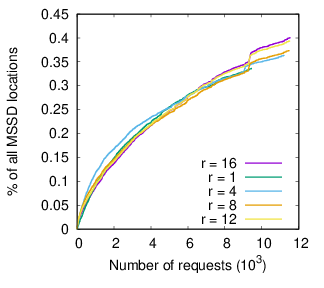}
    \caption{Google Places} 
   \label{fig:gpradius}
   
  \end{subfigure} 
  \begin{subfigure}[b]{0.3\linewidth}
    \centering
    \includegraphics[width=\linewidth]{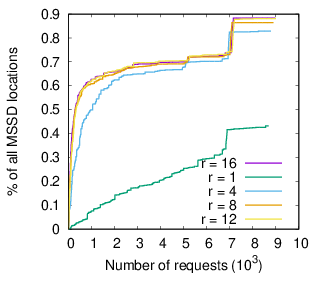}
      \caption{Foursquare} 
   \label{fig:fsqradius}
   
  \end{subfigure}
  
  \begin{subfigure}[b]{0.3\linewidth}
    \centering
    \includegraphics[width=\linewidth]{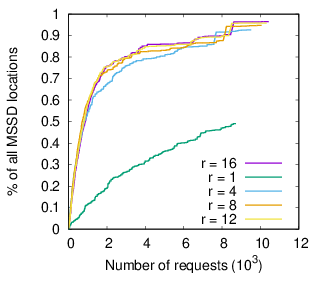}
      \caption{Flickr} 
   \label{fig:flickrradius}
   
  \end{subfigure} 
  \begin{subfigure}[b]{0.3\linewidth}
    \centering \includegraphics[width=\linewidth]{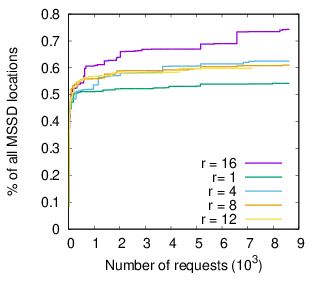}
      \caption{Twitter C} 
   \label{fig:twradius}
   
  \end{subfigure}\begin{subfigure}[b]{0.3\linewidth}
    \centering \includegraphics[width=\linewidth]{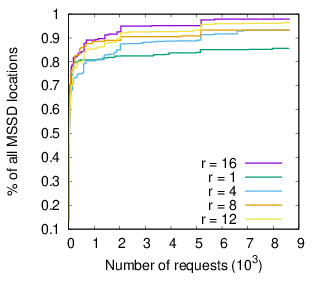}
      \caption{Twitter N} 
   \label{fig:ntwradius}
   
  \end{subfigure} 

  \caption{Requests vs locations for different $r$ in  \opt\ } 
  \label{fig:optradius}

\end{figure*}

\subsection{Requests vs locations using several seeds}

In this subsection, we study the number of requests and the number of locations retrieved for each \seed\ algorithm in all sources using different seeds. Each graph in Fig. \ref{fig:allinalltradeoff} shows the number of requests in $10^2$ (x-axis) and the number of locations in $10^2$ (y-axis) using a specific seed (in the right bottom corner of the graph) for each source (in the caption of the figure); e.g.  the results from experimenting with Flickr having Krak as seed are presented in the first graph. The best position in this representation is the left-top corner, which means fewer requests and more locations. It is obvious that \opt\ has saved this corner in most of the cases. Some remarkable results are noted when experimenting with Flickr using Yelp as a seed, where \opt\ gets twice the locations of \rec\ with  30\% for \rec\ requests. For Yelp using Flickr as seed and for Foursquare using GP as seed, \opt\ retrieves twice the locations of \rec\ with only 15\% of \rec\ requests. In Twitter, \opt\ is quite economic, using always less than 6\% of \rec\ requests and retrieving 96\% -170\% of its data. Google Places showed similar results for all seeds, except with Foursquare, where \opt\ is very careful with the requests (less than 10\%) but it retrieves 30\%-60\% of the locations of rec. For Google Places with Foursquare as seed, \opt\ retrieves 130\% of rec locations with only 8\% of the requests.

\begin{figure*} [htb]

\begin{subfigure}[b]{\linewidth}

   \includegraphics[width=\linewidth]{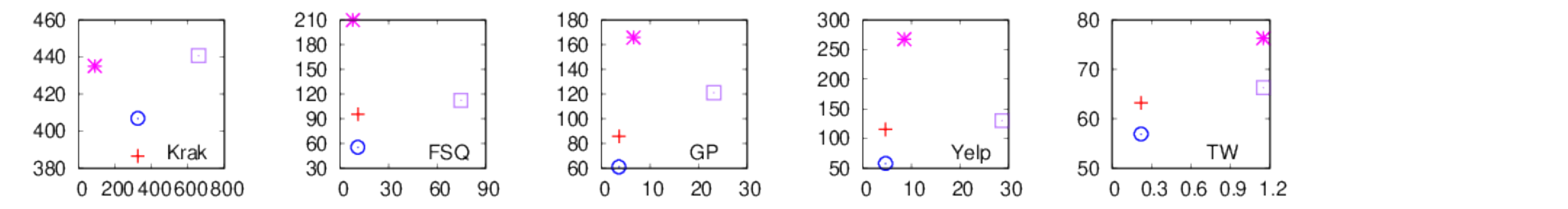}
  \caption{Flickr} 
   \label{fig:flickrss}
    
  \end{subfigure}

  \begin{subfigure}[b]{\linewidth}

   \includegraphics[width=\linewidth]{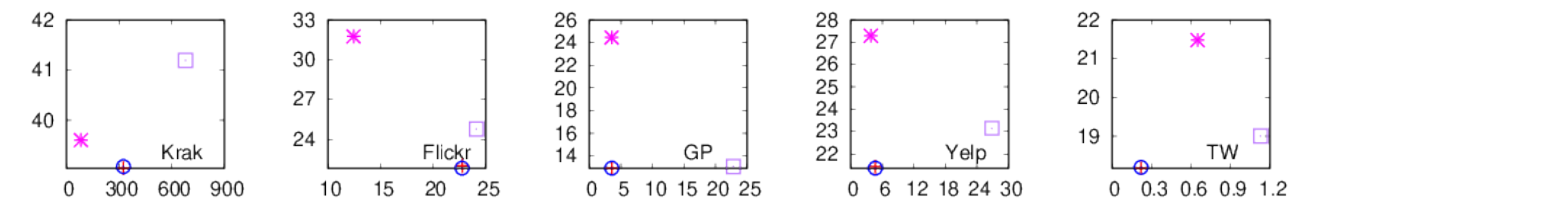}
   \caption{Foursquare} 
   \label{fig:fsqss}
    
  \end{subfigure}

  \begin{subfigure}[b]{\linewidth}
    \centering
   \includegraphics[width=\linewidth]{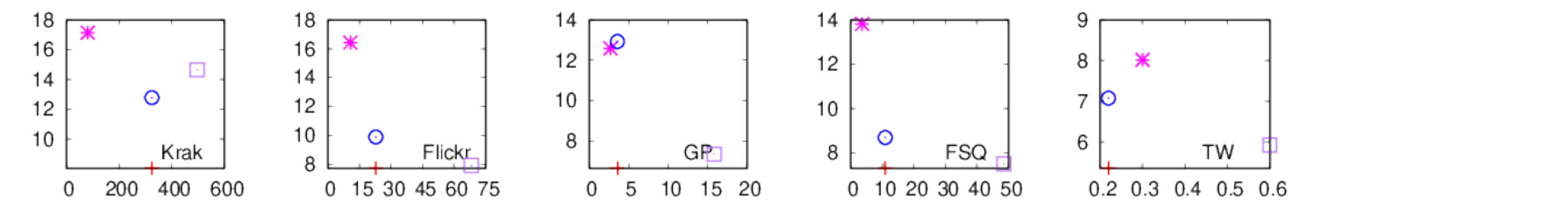}
   \caption{Yelp} 
   \label{fig:yelpss}
    
  \end{subfigure}
  
    \begin{subfigure}[b]{\linewidth}
    \centering
   \includegraphics[width=\linewidth]{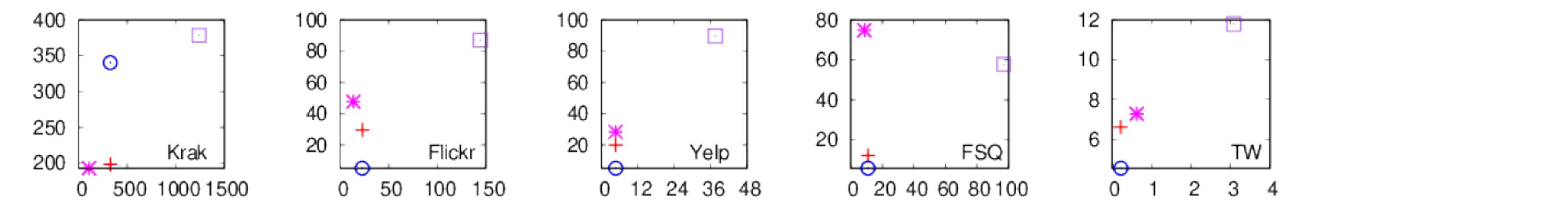}
   \caption{Google Places} 
   \label{fig:gpss}
    
  \end{subfigure}
  
    \begin{subfigure}[b]{\linewidth}
    \centering
   \includegraphics[width=\linewidth]{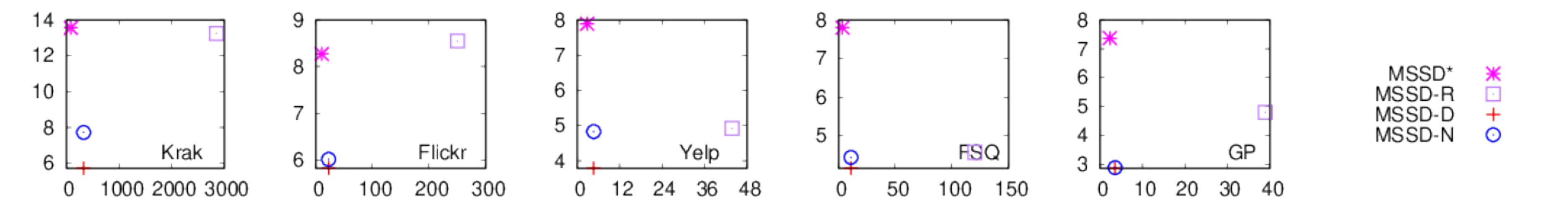}
   \caption{Twitter} 
   \label{fig:twss}
    
  \end{subfigure}
  
  \caption{Number of request  $10^2$ (x-axis) versus number of locations $10^2$ (y-axis) for all \seed\ algorithms resulting from experiments with (a)Flickr, (b)Foursquare, (c)Yelp, (d)Google Places and (e)Twitter using different seeds (noted in the graphs)}
  \label{fig:allinalltradeoff}

\end{figure*}


\begin{thebibliography}{10}

\bibitem{armenatzoglou2013general}
N.~Armenatzoglou, S.~Papadopoulos, and D.~Papadias.
\newblock A general framework for geo-social query processing.
\newblock {\em {PVLDB}}, 2013.

\bibitem{bennacer2017interpreting}
N.~Bennacer, F.~Bugiotti, M.~Hewasinghage, S.~Isaj, and G.~Quercini.
\newblock Interpreting reputation through frequent named entities in twitter.
\newblock In {\em WISE}, 2017.

\bibitem{burini2018urban}
F.~Burini, N.~Cortesi, K.~Gotti, and G.~Psaila.
\newblock The urban nexus approach for analyzing mobility in the smart city:
  Towards the identification of city users networking.
\newblock {\em Mobile Information Systems}, 2018.

\bibitem{cai2016using}
J.~L.~Z. Cai, M.~Yan, and Y.~Li.
\newblock Using crowdsourced data in location-based social networks to explore
  influence maximization.
\newblock In {\em {INFOCOM}}, 2016.

\bibitem{cheng2010you}
Z.~Cheng, J.~Caverlee, and K.~Lee.
\newblock You are where you tweet: a content-based approach to geo-locating
  twitter users.
\newblock In {\em CIKM}, 2010.

\bibitem{cho2011friendship}
E.~Cho, S.~A. Myers, and J.~Leskovec.
\newblock Friendship and mobility: user movement in location-based social
  networks.
\newblock In {\em {KDD}}, 2011.

\bibitem{crandall2010inferring}
D.~J. Crandall, L.~Backstrom, D.~Cosley, S.~Suri, D.~Huttenlocher, and
  J.~Kleinberg.
\newblock Inferring social ties from geographic coincidences.
\newblock {\em {PNAS}}, 2010.

\bibitem{emrich2014geo}
T.~Emrich, M.~Franzke, N.~Mamoulis, M.~Renz, and A.~Z{\"u}fle.
\newblock Geo-social skyline queries.
\newblock In {\em {DASFAA}}, 2014.

\bibitem{ester1996density}
M.~Ester, H.~Kriegel, J.~Sander, X.~Xu, et~al.
\newblock A density-based algorithm for discovering clusters in large spatial
  databases with noise.
\newblock In {\em {KDD}}, 1996.

\bibitem{feng2015personalized}
S.~Feng, X.~Li, Y.~Zeng, G.~Cong, Y.~M. Chee, and Q.~Yuan.
\newblock Personalized ranking metric embedding for next new poi
  recommendation.
\newblock In {\em IJCAI}, 2015.

\bibitem{ference2013location}
G.~Ference, M.~Ye, and W.~Lee.
\newblock Location recommendation for out-of-town users in location-based
  social networks.
\newblock In {\em {CIKM}}, 2013.

\bibitem{ferrari2011extracting}
L.~Ferrari, A.~Rosi, M.~Mamei, and F.~Zambonelli.
\newblock Extracting urban patterns from location-based social networks.
\newblock In {\em {GIS-LBSN}}, 2011.

\bibitem{gao2013exploring}
H.~Gao, J.~Tang, X.~Hu, and H.~Liu.
\newblock Exploring temporal effects for location recommendation on
  location-based social networks.
\newblock In {\em RecSys}, 2013.

\bibitem{gao2015content}
H.~Gao, J.~Tang, X.~Hu, and H.~Liu.
\newblock Content-aware point of interest recommendation on location-based
  social networks.
\newblock In {\em AAAI}, 2015.

\bibitem{hristova2016measuring}
D.~Hristova, M.~J. Williams, M.~Musolesi, P.~Panzarasa, and C.~Mascolo.
\newblock Measuring urban social diversity using interconnected geo-social
  networks.
\newblock In {\em {WWW}}, 2016.

\bibitem{jurgens2013s}
D.~Jurgens.
\newblock That's what friends are for: Inferring location in online social
  media platforms based on social relationships.
\newblock {\em ICWSM}, 2013.

\bibitem{jurgens2015geolocation}
D.~Jurgens, T.~Finethy, J.~McCorriston, Y.~T. Xu, and D.~Ruths.
\newblock Geolocation prediction in twitter using social networks: A critical
  analysis and review of current practice.
\newblock {\em {ICWSM}}, 2015.

\bibitem{kwak2010twitter}
H.~Kwak, C.~Lee, H.~Park, and S.~Moon.
\newblock What is twitter, a social network or a news media?
\newblock In {\em WWW}, 2010.

\bibitem{lee2010measuring}
R.~Lee and K.~Sumiya.
\newblock Measuring geographical regularities of crowd behaviors for
  twitter-based geo-social event detection.
\newblock In {\em {GIS-LBSN}}, 2010.

\bibitem{lee2006statistical}
S.~H. Lee, P.~Kim, and H.~Jeong.
\newblock Statistical properties of sampled networks.
\newblock {\em Physical Review E}, 2006.

\bibitem{li2014efficient}
G.~Li, S.~Chen, J.~Feng, K.~Tan, and W.~Li.
\newblock Efficient location-aware influence maximization.
\newblock In {\em {SIGMOD}}, 2014.

\bibitem{li2017geo}
J.~Li, T.~Sellis, J.~S. Culpepper, Z.~He, C.~Liu, and J.~Wang.
\newblock Geo-social influence spanning maximization.
\newblock {\em TKDE}, 2017.

\bibitem{li2013spatial}
L.~Li, M.~F. Goodchild, and B.~Xu.
\newblock Spatial, temporal, and socioeconomic patterns in the use of twitter
  and flickr.
\newblock {\em {CaGIS}}, 2013.

\bibitem{li2009analysis}
N.~Li and G.~Chen.
\newblock Analysis of a location-based social network.
\newblock In {\em {CSE}}, 2009.

\bibitem{ijcai2017-314}
A.~Likhyani, S.~Bedathur, and D.~P.
\newblock Locate: Influence quantification for location promotion in
  location-based social networks.
\newblock In {\em IJCAI-17}, 2017.

\bibitem{liu2013point}
B.~Liu and H.~Xiong.
\newblock Point-of-interest recommendation in location based social networks
  with topic and location awareness.
\newblock In {\em {SDM}}, 2013.

\bibitem{liu2017experimental}
Y.~Liu, T.-A.~N. Pham, G.~Cong, and Q.~Yuan.
\newblock An experimental evaluation of point-of-interest recommendation in
  location-based social networks.
\newblock {\em PVLDB}, 2017.

\bibitem{morstatter2013sample}
F.~Morstatter, J.~Pfeffer, H.~Liu, and K.~M. Carley.
\newblock Is the sample good enough? comparing data from twitter's streaming
  api with twitter's firehose.
\newblock In {\em {ICWSM}}, 2013.

\bibitem{nemhauser1978analysis}
G.~L. Nemhauser, L.~A. Wolsey, and M.~L. Fisher.
\newblock An analysis of approximations for maximizing submodular set
  functions—i.
\newblock {\em Mathematical programming}, 1978.

\bibitem{noulas2012random}
A.~Noulas, S.~Scellato, N.~Lathia, and C.~Mascolo.
\newblock A random walk around the city: New venue recommendation in
  location-based social networks.
\newblock In {\em {SocialCom/PASSAT }}, 2012.

\bibitem{preoctiuc2013mining}
D.~Preo{\c{t}}iuc-Pietro and T.~Cohn.
\newblock Mining user behaviours: a study of check-in patterns in location
  based social networks.
\newblock In {\em WebSci}, 2013.

\bibitem{saleem2018effective}
M.~A. Saleem, R.~Kumar, T.~Calders, and T.~B. Pedersen.
\newblock Effective and efficient location influence mining in location-based
  social networks.
\newblock {\em Knowledge and Information Systems}, 2018.

\bibitem{saleem2017location}
M.~A. Saleem, R.~Kumar, T.~Calders, X.~Xie, and T.~B. Pedersen.
\newblock Location influence in location-based social networks.
\newblock In {\em {WSDN}}, 2017.

\bibitem{scellato2010distance}
S.~Scellato, C.~Mascolo, M.~Musolesi, and V.~Latora.
\newblock Distance matters: Geo-social metrics for online social networks.
\newblock In {\em {WOSN}}, 2010.

\bibitem{scellato2011exploiting}
S.~Scellato, A.~Noulas, and C.~Mascolo.
\newblock Exploiting place features in link prediction on location-based social
  networks.
\newblock In {\em {KDD}}, 2011.

\bibitem{wang2013location}
H.~Wang, M.~Terrovitis, and N.~Mamoulis.
\newblock Location recommendation in location-based social networks using user
  check-in data.
\newblock In {\em {SIGSPATIAL/GIS}}, 2013.

\bibitem{weiler2015geo}
M.~Weiler, K.~A. Schmid, N.~Mamoulis, and M.~Renz.
\newblock Geo-social co-location mining.
\newblock In {\em GeoRich@SIGMOD}, 2015.

\bibitem{yao2016poi}
Z.~Yao, Y.~Fu, B.~Liu, Y.~Liu, and H.~Xiong.
\newblock Poi recommendation: A temporal matching between poi popularity and
  user regularity.
\newblock In {\em ICDM}, 2016.

\bibitem{yin2016discovering}
H.~Yin, Z.~Hu, X.~Zhou, H.~Wang, K.~Zheng, Q.~V.~H. Nguyen, and S.~Sadiq.
\newblock Discovering interpretable geo-social communities for user behavior
  prediction.
\newblock In {\em {ICDE}}, 2016.

\bibitem{yu2017friend}
F.~Yu, N.~Che, Z.~Li, K.~Li, and S.~Jiang.
\newblock Friend recommendation considering preference coverage in
  location-based social networks.
\newblock In {\em {PAKDD}}, 2017.

\bibitem{zhang2012evaluating}
C.~Zhang, L.~Shou, K.~Chen, G.~Chen, and Y.~Bei.
\newblock Evaluating geo-social influence in location-based social networks.
\newblock In {\em {CIKM}}, 2012.

\bibitem{zhang2013igslr}
J.~Zhang and C.~Chow.
\newblock igslr: personalized geo-social location recommendation: a kernel
  density estimation approach.
\newblock In {\em {SIGSPATIAL/GIS}}, 2013.

\bibitem{zhao2016stellar}
S.~Zhao, T.~Zhao, H.~Yang, M.~R. Lyu, and I.~King.
\newblock Stellar: Spatial-temporal latent ranking for successive
  point-of-interest recommendation.
\newblock In {\em AAAI}, 2016.

\bibitem{zhu2017geo}
Q.~Zhu, H.~Hu, C.~Xu, J.~Xu, and W.~Lee.
\newblock Geo-social group queries with minimum acquaintance constraints.
\newblock {\em {VLDB J.}}, 2017.

\end{thebibliography}
\end{document}